\documentclass[lettersize,journal]{IEEEtran}
\usepackage{amsmath,amsfonts}
\pdfoutput=1
\usepackage{algorithm}
\usepackage{array}
\usepackage{textcomp}
\usepackage{stfloats}
\usepackage{url}
\usepackage{verbatim}
\usepackage{graphicx}
\usepackage{cite}
\usepackage{amssymb}
\newcommand{\ourname}{\textit{NCSAC}}
\newcommand{\stitle}[1]{\vspace{1ex} \noindent{\bf #1}}

\usepackage{color}

\usepackage{bbding} 
\usepackage{bm}
\usepackage{amsmath}
\usepackage[noend]{algpseudocode}
\usepackage{algorithmicx,algorithm}
\usepackage{graphicx}
\usepackage{float}
\newtheorem{proof}{Proof}
\newtheorem{theorem}{Theorem}
\newtheorem{definition}{Definition}

\newtheorem{lemma}{Lemma}

\newtheorem{fact}{Fact}

\newtheorem{example}{Example}

\usepackage{balance}
\usepackage{subfigure}
\usepackage{multirow}
\usepackage{makecell}
\usepackage{array}
\usepackage{multicol}
\usepackage{booktabs} 

\begin{document}

\title{NCSAC: Effective Neural Community Search via Attribute-augmented Conductance}

	\author{Longlong Lin, Quanao Li, Miao Qiao, Zeli Wang, Jin Zhao, Rong-Hua Li,  Xin Luo,  Tao Jia
	\IEEEcompsocitemizethanks{\IEEEcompsocthanksitem Longlong Lin, Quanao Li, and Xin Luo are with the College of
		Computer and Information Science,
		Southwest University. Email: \{longlonglin; luoxin\}@swu.edu.cn, quanaoli174@gmail.com.
        	\IEEEcompsocthanksitem
            Miao Qiao is with the University of Auckland, New Zealand.  Email: miao.qiao@aucklanduni.ac.nz.
					\IEEEcompsocthanksitem  Zeli Wang is with Chongqing University of Posts and Telecommunications.
		Email: zlwang@cqupt.edu.cn.
        		\IEEEcompsocthanksitem  Jin Zhao is with Huazhong University of Science and Technology. Email: zjin@hust.edu.cn.
				\IEEEcompsocthanksitem  Rong-Hua Li is with 
		Beijing Institute of Technology.  Email: lironghuabit@126.com.\IEEEcompsocthanksitem 
        Tao Jia is with the College of
		Computer and Information Science, Chongqing Normal University and Southwest University. Email: tjia@swu.edu.cn.}

	\thanks{Manuscript received XXX, XXX; revised XXX, XXXX.}}

\markboth{Journal of \LaTeX\ Class Files,~Vol.~14, No.~8, August~2021}%
{Shell \MakeLowercase{\textit{et al.}}: A Sample Article Using IEEEtran.cls for IEEE Journals}


\maketitle

\begin{abstract}
Identifying locally dense communities closely connected to the user-initiated query node is crucial for a wide range of applications.
Existing approaches either solely depend on rule-based
constraints or exclusively utilize deep learning technologies to identify target communities. Therefore, an important question is proposed: can deep learning be integrated with rule-based constraints to elevate the quality of community search?  In this paper, we affirmatively address this question by introducing a novel approach called Neural Community Search via Attribute-augmented Conductance, abbreviated as NCSAC. Specifically, NCSAC first proposes a novel concept of attribute-augmented conductance, which harmoniously blends the  (internal and external) structural proximity and the attribute similarity. Then, NCSAC  extracts a coarse candidate community of satisfactory quality using the proposed attribute-augmented conductance. Subsequently, NCSAC frames the community search as a graph optimization task, refining the candidate community through sophisticated reinforcement learning techniques, thereby producing high-quality results. Extensive experiments on
six real-world graphs and ten competitors demonstrate the
superiority of our solutions in terms of accuracy, efficiency,
and scalability. Notably, the proposed solution outperforms state-of-the-art methods, achieving an impressive F1-score improvement ranging from 5.3\% to 42.4\%. For reproducibility purposes, the source code is available at https://github.com/longlonglin/ncsac.
\end{abstract}


\section{Introduction} \label{sec:intro}

A community is defined as a subgraph characterized by internal cohesiveness and external sparsity. Community Search (CS) aims to identify the specific community that exhibits a strong correlation with the user-initiated query nodes \cite{fang2020survey}. The importance of CS lies in its wide range of practical applications, including social e-commerce and uncovering personal backgrounds \cite{fang2020survey}.


Existing CS models can be generally classified into two categories (details in Section \ref{sec:existing}): rule-based CS methods \cite{k-core2,k-truss3,k-ecc,DMCS,DBLP:journals/pvldb/WuJLZ15,DBLP:conf/sigmod/DaiQC22,DBLP:conf/kdd/YeLLLLW24,DBLP:conf/focs/AndersenCL06,DBLP:conf/kdd/KlosterG14,DBLP:conf/aaai/LinLJ23,xie2024efficient,DBLP:journals/dase/LiZSYG24} and learning-based CS methods \cite{ICS-GNN,QD-GNN,coclep,cgnp,TransZero,CommunityAf,wang2024scalable,cai2023heterocs}. The former typically identifies the resultant community by utilizing predefined structural constraints (e.g.,  subgraph-based $k$-core/truss  \cite{k-core2,k-truss3} or optimization-based density/modularity\cite{DBLP:journals/pacmmod/ZhangLLZQW25,DMCS,DBLP:journals/pvldb/WuJLZ15,DBLP:conf/sigmod/DaiQC22,DBLP:conf/kdd/YeLLLLW24}). However, they impose cumbersome constraints on the target community, leading to structural inflexibility. For instance, some nodes within the ground-truth community (such as boundary nodes) may struggle to meet these constraints, causing users to question the model's accuracy \cite{ICS-GNN}. Fortunately, learning-based CS methods can relieve the inflexibility of rule-based CS methods. In particular, one popular methodology is to define the CS problem as a node classification task \cite{ICS-GNN,QD-GNN,coclep,cgnp,TransZero}. These methods employ a two-stage framework. Initially, the model is trained using semi-supervised or unsupervised techniques to generate node embedding vectors. Subsequently, these embedding vectors are utilized to evaluate the probability of each vertex belonging to the target community. Another methodology is to frame the CS problem as a generative task, exemplified by  CommunityAF \cite{CommunityAf}. Specifically, it first utilizes the Graph Neural Networks (GNNs) \cite{GCN} to derive node representations and subsequently employs autoregressive flow to generate the community from the query node. In contrast to node classification models, generative models offer greater flexibility and eliminate concerns regarding community connectivity \cite{CommunityAf}.


However, existing methods rely either entirely on rule-based constraints or solely leverage deep learning techniques for identifying target communities. Thus, a significant inquiry arises: how can deep learning be combined with traditional rule-based constraints to enhance the accuracy of CS? In this paper, we answer this question affirmatively.  Two significant challenges are involved:


\noindent\textit{Challenge 1: How to identify a coarse  candidate community of satisfactory quality?} The most straightforward method is to use the traditional rule-based CS models. However, this technique exhibits significant limitations. First, they only consider the internal cohesiveness and ignore the external sparsity, which is an important factor in measuring the quality of a community \cite{DBLP:journals/pacmmod/KamalB24,DBLP:journals/eswa/HeLYLJW24,DBLP:conf/www/LeskovecLM10}, resulting in unsatisfactory quality. Second, they only optimize the topology and ignore important node attributes, further limiting their performance. Thus, our first challenge is to identify a candidate community that accounts for both the (internal and external) structural proximity and attribute similarity.

\noindent\textit{Challenge 2: How to refine the candidate community?} After identifying the candidate community, it still lacks flexibility and robustness, as it remains a predefined subgraph. Recently, \emph{ALICE} \cite{ALICE} and \emph{TransZero} \cite{TransZero} have used GNNs to refine the candidate community. However, they treat the candidate community as a new graph $G_{new}$ and apply GNNs on $G_{new}$ for node classification tasks. As a result, their methods are heavily dependent on the performance of $G_{new}$, since nodes outside $G_{new}$ have no chance to contribute to the target community, resulting in poor quality (Section \ref{sec:experiments}). Therefore, refining these outside and inside nodes of $G_{new}$ to enhance the candidate community becomes our second challenge.


\stitle{Our Contributions.} To address these challenges, we propose a novel neural community search model via attribute-augmented conductance (named \ourname) which consists of two principal components: the candidate community extraction module and the neural community refinement module. The former leverages the newly-proposed \emph{attributed-augmented conductance} to identify a candidate community of satisfactory quality. Following this, the community refinement module non-trivially utilizes powerful reinforcement learning \cite{ppo1,ppo2,tpro} to iteratively add and remove nodes starting from the candidate community, thereby refining the quality of the candidate community. More concretely, to address Challenge 1,  we resort to the well-known conductance \cite{DBLP:conf/focs/AndersenCL06,DBLP:conf/kdd/KlosterG14,DBLP:conf/www/LeskovecLM10,hk2} to capture the internal cohesion and external sparsity. However, the traditional conductance only models the structural information while neglecting crucial attribute information. To overcome this dilemma, we propose a novel attributed-augmented conductance metric, which seamlessly integrates the topology structure and node attributes. Therefore, within the candidate community extraction module, we initiate the process from the given query node to determine a candidate community with the (most favorable) approximate attribute-augmented conductance. 

For Challenge 2, we design a neural community refinement module that frames the refinement process as a \textit{combinatorial graph optimization} problem (Definition \ref{def:cr}). The core idea is to improve the community quality by removing irrelevant nodes or appending new ones from the community boundaries, optimizing our objective through careful node selection. Unfortunately, the community refinement problem is NP-hard (Theorem \ref{thm:rl}), rendering direct solutions impractical. To address this, we present an efficient community refiner based on reinforcement learning.
Specifically, we cast this problem as a reinforcement learning task and employ the Proximal Policy Optimization (PPO) \cite{ppo1,ppo2} technique for exploration. Besides, we design two distinct policy networks: one for predicting node additions and another for predicting node deletions. Through reinforcement learning, our approach can effectively refine the candidate community. On top of that, we introduce a termination strategy to ensure that the final community structure remains flexible and adaptable.

We conduct extensive experiments on six real-life graphs and ten competitors to evaluate the effectiveness and scalability of the proposed solutions. The empirical results show that our proposed solution can achieve an F1-score improvement ranging from 5.3\% to 42.4\% with comparable runtime in comparison to baselines.


\section{Preliminaries} \label{sec:pro}

\subsection{Problem Definition}
Consider an undirected attributed graph $G=(V, E, \mathcal{F})$, where $V$ is a node set of cardinality $n$, $E$ is an edge set of cardinality $m$, and $\mathcal{F} \in R^{n\times k}$ is an attribute (feature) matrix. In particular, $\mathcal{F}_{uf}=1 $ indicates that node $u$ possesses the attribute $f$, while $\mathcal{F}_{uf}=0 $ signifies that node $u$ does not contain the attribute $f$. We refer to $\mathcal{F}[u] \in R^{1\times k}$ as the attribute vector of node $u$. Let $N(u)=\{v| (u,v) \in E\}$
be the neighbors of node $u$ and $d(u)=|N(u)|$ be the degree of $u$. We use $\boldsymbol{A} \in R^{n \times n}$ to denote the adjacency matrix of $G$, where $\boldsymbol{A}_{uv}=1$ if $(u,v) \in E$, and $\boldsymbol{A}_{uv}=0$ otherwise. Besides, we use $\boldsymbol{D} \in R^{n \times n}$ to represent the degree diagonal matrix of $G$, in which $\boldsymbol{D}_{uu}=d(u)$. Table \ref{table 1} provides the frequently used symbols. 

\begin{table}[t]
\caption{The frequently used symbols.}
	\label{table 1}
    \centering\footnotesize
\begin{tabular}{l|l}
\toprule
\textbf{Symbol} & \textbf{Definitions} \\ 
\midrule
$G=(V,E,\mathcal{F})$ & an undirected graph with the node attribute matrix $\mathcal{F}$ \\
$q$, $C_{q}$ & the query node, the community containing $q$ \\
$\phi_{t}(.)$, $\phi_{a}(.)$& topology-based conductance, attribute-based conductance\\
$\Phi(.)$ & attribute-augmented conductance\\
$L_C$, $L_T$, $\mathcal{L}(\phi)$  &the contrastive loss, triple loss, PPO loss \\
$\mathcal{S},\mathcal{A},\mathcal{P},\mathcal{R},\gamma$  &State, Action, State Transition, Reward, discount factor \\
\bottomrule
\end{tabular}
\end{table}


\begin{definition} [Community Search \cite{k-core2,k-truss3,k-ecc,DMCS,DBLP:journals/pvldb/WuJLZ15,DBLP:conf/sigmod/DaiQC22,DBLP:conf/kdd/YeLLLLW24,DBLP:conf/focs/AndersenCL06,DBLP:conf/kdd/KlosterG14,DBLP:conf/aaai/LinLJ23}] \label{def:cs}
Given an undirected attribute graph $G=(V, E, \mathcal{F})$, a score function $g(.)$ (e.g., minimum degree or conductance), and a query node $q\in V$,  the goal of community search is to identify a vertex set  $C_{q}\subseteq V$ such that (1) $q\in C_{q}$; (2) $C_{q}$ is connected; (3) $g(C_{q})$ is optimal.
\end{definition}

\begin{table}[t]
\caption{A comparison of state-of-the-art community search methods. "SL" represents Supervised Learning, "UL" indicates Unsupervised Learning, "RL" indicates  Reinforcement Learning,  "NC" represents Node Classification, "GR" represents Generation, and "GO" indicates Graph Optimization.}\vspace{-0.3cm}
\centering\footnotesize
\begin{tabular}{lccc}
\toprule 
Methodologies  & \makecell{Internal Cohesiveness\\ \&External Sparsity}& \makecell{Learning\\ Paradigm} & \makecell{Problem\\ Definition} \\
\midrule
Subgraph-based \cite{k-core2,k-truss3,k-ecc} &  $\times$  & $\times$ & $N/A$\\
Density \cite{DBLP:journals/pvldb/WuJLZ15,DBLP:conf/sigmod/DaiQC22,DBLP:conf/kdd/YeLLLLW24} &  $\times$ & $\times$ & $N/A$\\
Modularity \cite{DMCS} &  $\times$  & $\times$ & $N/A$\\
Conductance \cite{DBLP:conf/focs/AndersenCL06,DBLP:conf/kdd/KlosterG14,DBLP:conf/aaai/LinLJ23} &  \Checkmark  & $\times$ & $N/A$\\
\midrule
\emph{ICS-GNN} \cite{ICS-GNN} & $\times$  & SL & NC \\
\emph{QD-GNN} \cite{QD-GNN}  & $\times$  & SL & NC \\
\emph{COCLEP} \cite{coclep}  & $\times$ &  Semi-SL & NC \\
\emph{CGNP} \cite{cgnp}  &  $\times$ &  Semi-SL & NC \\
\emph{TransZero} \cite{TransZero}  & $\times$ &  UL & NC \\
\emph{ALICE} \cite{ALICE} & $\times$ &  SL & NC \\ 
\emph{CommunityAF} \cite{CommunityAf}  & $\times$ &  SL & GR \\
\midrule
\emph{\ourname} (This paper) & \Checkmark &  RL & GO \\
\bottomrule
\end{tabular}\vspace{-0.5cm}
\label{table:cs_comparison}
\end{table}

In this paper, we conceptualize community search as a \emph{graph optimization problem} (Definition \ref{def:cr}) rather than as a binary classification task, as seen in existing approaches like \emph{ICS-GNN} \cite{ICS-GNN} and \emph{TransZero} \cite{TransZero}. At a high level, we first identify a coarse yet high-quality candidate community through a powerful subgraph extraction process, laying the groundwork for subsequent refinement. Building upon the extracted candidate community, we systematically refine it by incorporating additional relevant nodes while removing extraneous ones. This iterative refinement process is designed to optimize community quality, ensuring a more accurate representation of the underlying relationships within the graph. Note that attribute community search (i.e., both attributes and nodes serve as query requests, e.g., \emph{AQD-GNN} \cite{QD-GNN}) and multiple query nodes are orthogonal to ours and fall outside the scope of this paper, presenting potential directions for future research. Table \ref{table:cs_comparison} summarizes the state-of-the-art community search methods. 

\subsection{Graph Neural Network}
Graph Neural Network (GNN) is a specialized neural network structure tailored for processing graph data \cite{GCN,Graphsage,GAT}. It's engineered to encode network nodes as low-dimensional vectors, preserving both the network topology and node feature information, thus facilitating downstream tasks such as node classification and graph clustering \cite{GCN,DBLP:conf/mir/YuLLWOJ24}. GNN governs the information propagation and update process by defining node aggregation and update functions, enabling nodes to assimilate information from their neighbors and refine their representations. This iterative process continues until the model converges: $h_{u}^{(l+1)} =\mathrm{UPDATE}(h_{u}^{(l) },\mathrm{AGGREGATE}(h_{v}^{(l) }|v\in N(u)))$, in which 
$h_{u}^{(l) }$ indicates the representation of vertex $u$ in layer $l$. $\mathrm{AGGREGATE}(\cdot )$ is defined as a node aggregation function that weighs the representation of vertex $v$ w.r.t. vertex $u$. $\mathrm{UPDATE}(\cdot )$ is an update function that iteratively refines the features of all nodes  based on the aggregated information. 

\subsection{Reinforcement Learning} \label{subsec:rl}
Reinforcement learning (RL for short) is used to address the challenge of enabling agents to learn strategies that maximize returns or achieve specific goals through interactions with their environment \cite{ppo1,ppo2,tpro,lei2020reinforcement,deng2021unified,zheng2023dream}.  Typically, the interaction between an agent and its environment is modeled using a Markov Decision Process (MDP), which can be represented by a quintuple $(\mathcal{S},\mathcal{A},\mathcal{P},\mathcal{R},\gamma)$. Here, $\mathcal{S}$ denotes the state space, while $\mathcal{A}$ represents the action space, encompassing the set of viable actions within a given environment. $\mathcal{P}$ signifies the state transition function, with $\mathcal{P}(s'|s,a)$ denoting the probability of transitioning from state $s$ to state $s'$ upon taking action $a$. $\mathcal{R}$ represents the reward function, typically a deterministic function crafted by humans. $\gamma\in [0,1]$ stands for the discount factor, which is a constant. The goal of RL is to learn a policy by maximizing the expected \textit{discounted return} $U_{t}=E[\sum_{i=t}^{\infty}\gamma ^{i-t}\mathcal{R}_{i}]$, in which $\mathcal{R}_{i}$ is the reward for transitioning from state $s_{i-1}$ to $s_{i}$.

\begin{figure*}[t]
    \centering
\includegraphics[width=0.9\linewidth]{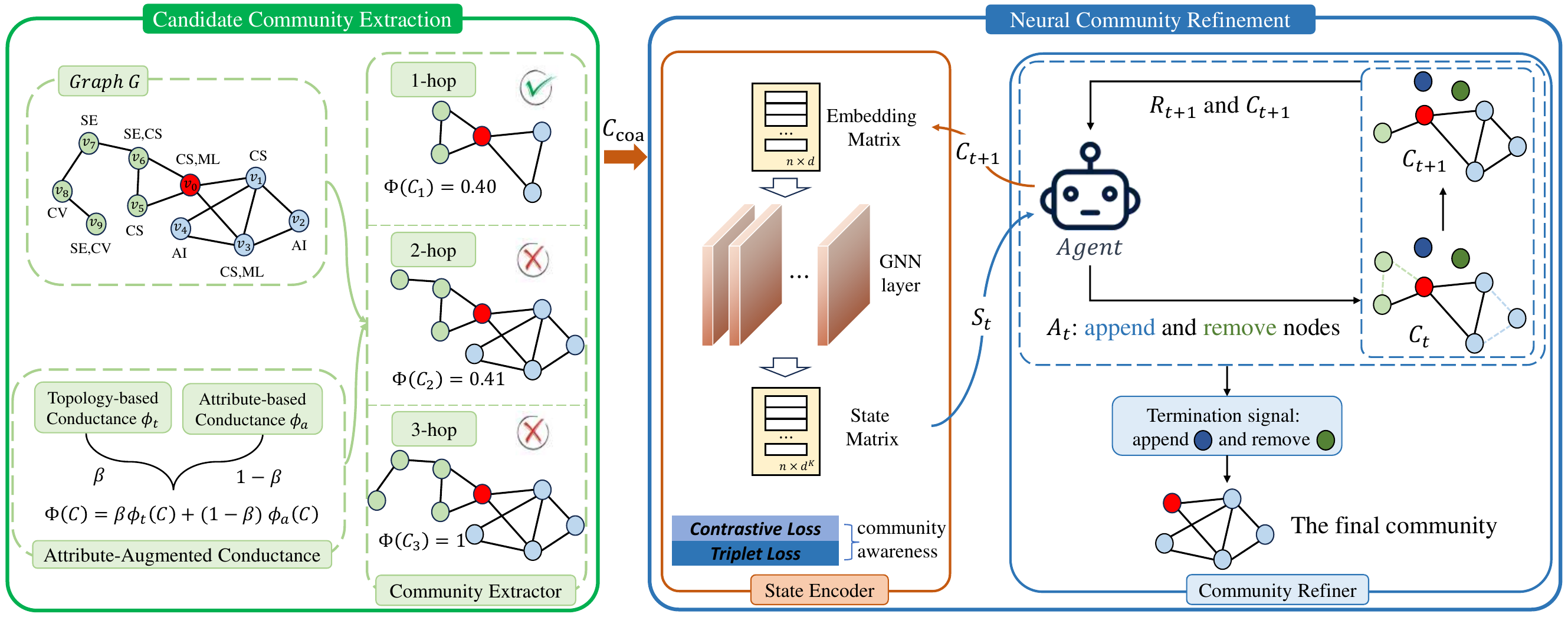}\vspace{-0.3cm}
    \caption{The framework of our proposed \ourname. In the candidate community extraction stage,  \ourname \ systematically extracts a coarse candidate community $C_{coa}$ by leveraging our novel attribute-augmented conductance ( the red node is the query node). Subsequently, we pre-train a state encoder to effectively integrate community awareness information into the state encoding. Finally,  we apply non-trivial reinforcement learning techniques to flexibly refine the coarse candidate community, ultimately yielding high-quality results (best view in color).}\vspace{-0.3cm}
    \label{fig:framework}
\end{figure*}

\section{The Proposed Framework} \label{sec:our}
In this section, we propose an effective community search method, \ourname, as illustrated in  Figure \ref{fig:framework}. Specifically, \ourname \ first utilizes our proposed \emph{attribute-augmented conductance} to adaptively extract a coarse candidate community by identifying the $h$-hop subgraph with optimal attribute-augmented conductance. We highlight that the attribute-augmented conductance simultaneously captures both the topological structure (i.e., internal cohesiveness and external sparsity) and the attributes of nodes. Consequently, the resulting coarse candidate community exhibits comparable community quality (Section \ref{ace}). Subsequently, \ourname \ pre-trains a graph neural network to serve as the state encoder, integrating community awareness into the state representation through our proposed contrastive loss and triplet loss. Once the coarse candidate community is encoded into the initial state by the state encoder, reinforcement learning is employed to initiate the community refinement process. At each step of refinement, the agent processes the community's state encoding and refines the community by incorporating promising nodes while simultaneously excluding noisy or irrelevant ones. The refined community is then re-encoded by the state encoder to generate the next state, enabling the agent to proceed with the subsequent refinement step. This process continues until the predefined termination strategy is satisfied (Section \ref{fcr}).



\section{Candidate Community Extraction} \label{ace}
This section first proposes a novel concept of attribute-augmented conductance to measure the quality of a community, which seamlessly blends structural and attribute information. Then, we devise an efficient dynamic update algorithm to adaptively identify high-quality candidate communities.

\subsection{Attribute-augmented  Conductance}
As previously mentioned, a crucial aspect of community search is the development of an indicator to measure community quality. To this end, we adopt the well-known free-parameter community metric, \emph{conductance}, due to its advantageous structural properties and solid theoretical foundation \cite{DBLP:conf/focs/AndersenCL06,DBLP:conf/kdd/KlosterG14,DBLP:conf/www/LeskovecLM10,hk2,DBLP:conf/aaai/LinLJ23,DBLP:journals/eswa/HeLYLJW24,DBLP:conf/kdd/Lin0WZL24,DBLP:journals/tkde/LinWLLQZ25}. Unlike metrics such as $k$-core \cite{k-core1}, $k$-truss \cite{k-truss1}, $k$-clique \cite{k-clique1}, and modularity \cite{modularity1}, which measure only the internal cohesion of a community, conductance takes into account both internal cohesion and external sparsity. This dual consideration makes conductance a more robust measure of community quality \cite{DBLP:conf/www/LeskovecLM10,11077407,DBLP:journals/pacmmod/KamalB24}.

\begin{definition} [Conductance] \label{def:co}
	Given a graph $G$ and a vertex set $C$, the conductance of $C$ is defined as	follows.
 \begin{equation} \label{conductance}
    \phi (C)=\frac{ \left | \mathrm{cut}(C,\bar{C}) \right | }{\min \{\mathrm{vol}(C),\mathrm{vol}(\bar{C}) \}}.
\end{equation}
\end{definition}
where $\bar{C}$ denotes the complement of the community, defined as $\bar{C} = V \backslash C$. The cut between community $C$ and its complement is represented as $\mathrm{cut}(C, \bar{C}) = \{(u, v) \mid u \in C, v \in \bar{C}\}$, while $\mathrm{vol}(C) = \sum_{u \in C} d(u)$ indicates the total sum of node degrees within community $C$. Therefore, a smaller value of $\phi(C)$ suggests that the number of edges exiting community $C$ is relatively low compared to the number of edges contained within it. Consequently, a lower value of $\phi(C)$ indicates a higher quality of community $C$ \cite{DBLP:conf/focs/AndersenCL06,DBLP:conf/kdd/KlosterG14,DBLP:conf/www/LeskovecLM10,hk2}. A more interesting and useful probabilistic explanation of the conductance is stated as follows.


\begin{theorem} \label{thm:esc}
	 Given a graph $G$ and a community $C$. Consider a random walk $(w_t)_{t\in \mathbb{N}}$ generated by the transition matrix $\boldsymbol{P}(=\boldsymbol{D}^{-1}\boldsymbol{A})$, and the initial state $w_0$ is in $C$ (or $\bar{C}$) and is randomly chosen following the degree distribution. We have 
	 \begin{equation}\small
	 \phi(C)=max\{Pr(w_1 \in \bar{C}| w_0 \in C), Pr(w_1 \in C| w_0 \in \bar{C})\}.
	 \end{equation}
\end{theorem}

\begin{proof}
Firstly, according to Eq.  ~\eqref{conductance}, we can deduce that
	 $\phi (C)=\frac{ \left | \mathrm{cut}(C,\bar{C}) \right | }{\min\{\mathrm{vol}(C),\mathrm{vol}(\bar{C})\}} = \max\{\frac{ \left | \mathrm{cut}(C,\bar{C}) \right | }{\mathrm{vol}(C)},\frac{ \left | \mathrm{cut}(C,\bar{C}) \right | }{\mathrm{vol}(\bar{C})} \}$.
Then, we have $Pr(w_1 \in \bar{C}| w_0 \in C) = \sum\limits_{u \in C}\sum\limits_{v \in \bar{C}}\frac{d(u)}{\mathrm{vol}(C)}\cdot \boldsymbol{P}_{uv}$ $=\sum\limits_{u \in C}\sum\limits_{v \in \bar{C}\cap N(u)}\frac{d(u)}{\mathrm{vol}(C)}\cdot \frac{1}{d(u)} $
$=\sum\limits_{u \in C}\sum\limits_{v \in \bar{C}\cap N(u)}\frac{1}{\mathrm{vol}(C)}$
$= \frac{ \left | \mathrm{cut}(C,\bar{C}) \right | }{\mathrm{vol}(C)}$. Here, $\frac{d(u)}{\mathrm{vol}(C)}$ denotes the probability of randomly selecting node $u$ from community $C$ following the degree distribution, while $\boldsymbol{P}_{uv}$ represents the probability that node $u$ transitions to node $v$ by one-hop random walk. 
Analogously, we have $Pr(w_1 \in C| w_0 \in \bar{C}) = \frac{ \left | \mathrm{cut}(C,\bar{C}) \right | }{\mathrm{vol}(\bar{C})}$.
Consequently, the theorem is proved. \noindent\hspace*{\fill}$\blacksquare$
\end{proof}

Theorem \ref{thm:esc} offers a random walk interpretation of conductance. 
Specifically, the conductance of the community $C$ is the probability that a random walker, beginning within $C$ (or $\bar{C}$), will escape to the outside by one-hop random walk. To distinguish it from the attribute-based  conductance discussed later, we redefine conductance as topology-based conductance, denoted by $\phi_t(C)$. 
However,  $\phi_t(C)$ falls short in this context, as it fails to consider the attributes associated with the nodes. Inspired by the random walk interpretation of conductance, we propose the attribute-based conductance to capture structural proximity and attribute similarity.

\begin{definition} [Attribute Edge] \label{def:ae}
Given an attributed graph $G=(V, E, \mathcal{F})$, the connection between nodes $u$ and $v$ through a specific attribute $f$ is considered an \textit{attribute edge}, denoted as $(u, v, f)$. Namely, $(u, v, f)$ indicates nodes $u$ and $v$ share attribute $f$ (i.e., , $\mathcal{F}_{uf}$=1 and  $\mathcal{F}_{vf}$=1).
\end{definition}

\begin{definition} [Attribute Degree] \label{def:ad}
Given an attributed graph $G=(V, E, \mathcal{F})$, the attribute degree of a node $u$, denoted as $d_{a}(u)$, is defined as the number of attribute edges connected to $u$. The formula for its calculation is presented as follows:
\begin{equation}\label{eq:ad}
d_{a}(u)=\sum_{v\in V\backslash \{u\}}^{} \mathcal{F}[u]\cdot \mathcal{F}[v]^\top.
\end{equation} 
\end{definition}

Definitions \ref{def:ae} and \ref{def:ad} map the original attributed graph  into a multigraph (i.e., two vertices can be connected by more than one edge), as depicted in Figure \ref{fig:example}. In this multigraph, the edges correspond to attribute edges, and the degree of each node reflects its attribute degree. Unfortunately, the multigraph may be extremely dense (even close to quadratic), especially if there is an attribute value shared by a decent portion, e..g, 50\%, of the nodes in the graph. Section \ref{subsec:ce} will show how to quickly obtain attribute-augmented
conductance without materializing the multigraph.


\begin{definition} [Attribute-based Random Walk] \label{def:arw}
Given an attribute graph $G$, the attribute-based transition probability from node $u$ to node $v$ is defined as follows:
\begin{equation}\label{eq:arw}
P^a_{uv}=\frac{\mathcal{F}[u] \cdot \mathcal{F}[v]^\top}{d_{a}(u)}.
\end{equation} 
\end{definition}

\begin{figure}[t]
    \centering
\includegraphics[width=0.9\linewidth]{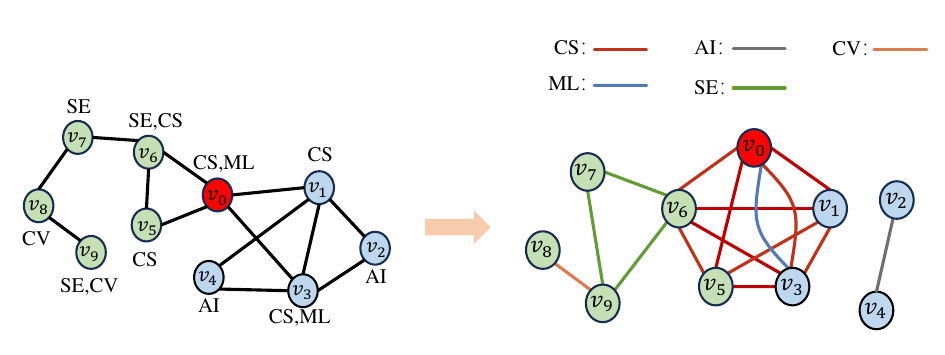}
    \vspace{-0.5cm}
    \caption{From the original attribute graph to the multigraph.}
    \label{fig:example}
\vspace{-5mm}
\end{figure} 

Intuitively, $P^a_{uv}$ models the transition probability from $u$ to $v$ based on the attribute set of $G$.

\begin{example}
Consider the example of nodes $v_{0}$ and $v_{3}$ in Figure \ref{fig:example}. $v_{0}$ is connected to $v_{3}$ through two attribute edges, specifically $(v_{0},v_{3},CS)$ and $(v_{0},v_{3},ML)$. Besides, $v_{0}$ also has three other attribute edges, namely, $(v_{0},v_{1},CS)$, $(v_{0},v_{5},CS)$, and $(v_{0},v_{6},CS)$. Thus, the attribute degree of node $v_{0}$ equals 5 and  the attribute transition probability from $v_{0}$ to $v_{3}$ is $P^{a}_{v_{0}v_{3}}=\frac{2}{5}=0.4$.
\end{example}

\begin{definition} [Attribute-based Conductance] \label{def:abc}
Given an attributed graph $G$ and a community \ $C$, the attribute-based conductance, denoted as $\phi_a (C)$, is defined as follows:
\begin{equation} \label{eq:abc}
\phi_a (C)=\frac{ \left | \mathrm{cut}_a(C,\bar{C}) \right | }{\min \{\mathrm{vol}_a(C),\mathrm{vol}_a(\bar{C}) \}}.
\end{equation}
Here, $\mathrm{cut}_a(C,\bar{C})=\{(v_i,v_j,f)\mid v_i\in C, v_j\in \bar{C}\}$ denotes the number of attribute edges connecting community $C$ and $\bar{C}$, while $\mathrm{vol}_a(C)=\sum_{u\in C}^{}d_a(u)$ represents the total sum of attribute degrees within $C$.
\end{definition}

We highlight that while the definitions of attribute-based conductance \(\phi_a (C)\) and topology-based conductance \(\phi_t (C)\) are similar, computing \(\phi_a (C)\) is significantly more complex (details in Section \ref{subsec:ce}).  Besides, analogous to Theorem \ref{thm:esc}, we can also understand attribute-based conductance from the perspective of attribute-based random walk, thereby enhancing its practical applicability. Specifically, we reformulate   Eq.~\eqref{eq:abc} as $\phi_a (C)=max\{Pr(\omega_1 \in \bar{C}| \omega_0 \in C), Pr(\omega_1 \in C| \omega_0 \in \bar{C})\}$, where $(\omega_t)_{t\in \mathbb{N}}$ represents a random walk generated by the attribute-based transition matrix $P^a$ (i.e., Eq. \ref{eq:arw}). Therefore, a smaller \(\phi_a (C)\) indicates a lower probability of transitioning from community \(C\) (or \(\bar{C}\)) to \(\bar{C}\) (or \(C\)) through the one-hop attribute-based random walk. Intuitively, this suggests that the nodes within the community largely share similar attributes, while there is limited overlap with the attributes of nodes outside the community. This observation aligns with our principles of internal attribute cohesiveness and external attribute sparsity.


Building on the concepts discussed earlier, we now introduce the  definition of attribute-augmented  conductance, which can capture simultaneously the topology and the attributes of the community.

\begin{definition} [Attribute-augmented  Conductance] \label{def:aec}
Given an attributed graph $G$, a community $C$, and a  preference parameter $\beta$,  the attribute-augmented  conductance is as follows:
\begin{equation} \label{eq:aec}
\Phi(C)=\beta \cdot \phi_t (C) + (1-\beta) \cdot \phi_a (C).
\end{equation}
\end{definition}

\stitle{Remark.} $\Phi(C)$ is constructed as a combination of $\phi_t (C)$ and $\phi_a (C)$ weighted by preference parameter $\beta$. This formulation captures a trade-off between topological structure and node attributes, allowing us to identify communities with distinct characteristics by varying $\beta$. In particular, when $\beta>0.5$, the community emphasizes the interaction edges among vertices. Conversely, when $\beta<0.5$, the focus shifts toward the attribute similarity among vertices.

\subsection{Adaptive Community Extractor} \label{subsec:ce}



Local search \cite{huang2014querying,k-core2,yao2021efficient,k-clique1,DMCS} is a well-established technique for identifying the community to which a query node belongs, utilizing an iterative strategy that incrementally adds nodes in a designed ordering. We emphasize that directly obtaining the optimal solution through this ordering is generally NP-hard \cite{k-core2}, and it becomes even more complex when considering conductance \cite{DBLP:conf/aaai/LinLJ23}, ultimately resulting in poor scalability. Fortunately, our principal aim is to efficiently derive a coarse candidate community with satisfactory quality, which will subsequently be refined through our proposed novel and non-trivial reinforcement learning techniques. To this end, we let $N^{h}(q)=\{v| dist(q,v) \leq h\}$ be the $h$-hop neighbors of the query node $q$, where $dist(q,v)$ represents the shortest-path distance between nodes $q$ and $v$. We designate $N^{h}(q)$ as the  candidate community, which has been well-supported by existing literature \cite{huang2014querying,yao2021efficient,k-clique1,DMCS}. Intuitively, nodes that are closer to the query node are more likely to belong to the target community. This approach implicitly defines the order of node addition (i.e., by adding based on distance in batches). 

\begin{example}
In Figure \ref{fig:framework}, we use $v_0$ as the query node, representing 1-hop, 2-hop, and 3-hop subgraphs as $C_1$, $C_2$, and $C_3$, respectively. We first compute the topology-based conductance values, obtaining $\phi_t(C_1)=0.55$, $\phi_t(C_2)=0.33$, and $\phi_t(C_3)=0.55$. From these values, we select $C_2$ as the coarse community by only using the topology-based conductance. Next, when incorporating attributes, the attribute-based conductance values are calculated as $\phi_a(C_1)=0.25$, $\phi_a(C_2)=0.5$, and $\phi_a(C_3)=1$. By setting $\beta=0.5$, indicating equal weighting of both topological structure and node attribute, the attribute-augmented  conductance values are adjusted to $\Phi(C_1)=0.40$, $\Phi(C_2)=0.415$, and $\Phi(C_3)=0.775$. In this case, $C_1$ becomes the coarse community, as it better aligns with the community definition compared to $C_2$.
\end{example}

\stitle {A failed attempt.}  According to Definition \ref{def:aec}, a naive approach to calculating attribute-augmented conductance is first to construct a multigraph (Figure \ref{fig:example}). Subsequently, it uses Eq. \ref{eq:abc} to obtain the attribute-based conductance and then applies Eq. \ref{eq:aec} to derive the attribute-augmented conductance. However, this naive solution has two-fold drawbacks: (1)  Materializing the multigraph requires prohibitively high time and space complexities. In particular, it takes $O(kn^2)$ time and $O(kn^2)$ space to materialize the multigraph by Eq. \ref{eq:ad}, resulting in poor scalability. Where $n$ and $k$ are the number of nodes and attributes of $G$, respectively. (2) When the subgraph $N^{h}(q)$ extends to $N^{h+1}(q)$, calculating the attribute-augmented  conductance of $N^{h+1}(q)$ from scratch entails significant redundant computations. Consequently, 
the naive method cannot deal with massive graphs with millions of edges (details in Section \ref{subsec:runtime}). Thus,  an important challenge is quickly calculating and updating attribute-augmented conductance without materializing the multigraph.

\stitle{Several in-depth observations.} Analysis of Eq. \ref{eq:aec} reveals that the main bottlenecks in calculating and updating attribute-augmented  conductance lie in the efficient computation of $\mathrm{vol}_a(C)$ and the rapid update of $\left | \mathrm{cut}_a(C,\bar{C}) \right |$. To overcome this challenge, we observe that a simple matrix-vector product can obtain the attribute degrees of all nodes. Besides, we  can also quickly and dynamically update  $| \mathrm{cut}_a(C,\bar{C})|$ by simply maintaining several arrays. For convenience, we state these observations as the following lemmas.

\begin{lemma}\label{lem:att_vol}
Consider an attributed graph $G=(V, E, \mathcal{F})$, for any node $u \in V$, the attribute degree of $u$ can be reformatted as $d_{a}(u)=\mathcal{F}[u] \cdot (\mathbf{1}_{n}\mathcal{F}-\mathbf{1}_{k})^\top$, in which $\mathbf{1}_{n} \in R^{1 \times n}$ (or $\mathbf{1}_{k} \in R^{1 \times k}$) is a vector with all elements equal to 1. Besides, we have $d_{a}=\mathcal{F} \cdot (\mathbf{1}_{n}\mathcal{F}-\mathbf{1}_{k})^\top$.
\end{lemma}

\begin{proof}
    By Eq. \ref{eq:ad}, we have  $d_{a}(u)=\sum_{v\in V\backslash \{u\}}^{} \mathcal{F}[u]\cdot \mathcal{F}[v]^\top=\mathcal{F}[u]\cdot\sum_{v\in V\backslash \{u\}}^{}  \mathcal{F}[v]^\top=\mathcal{F}[u] \cdot ((\mathbf{1}_{n}\mathcal{F})^\top-\mathcal{F}[u]^\top)=\mathcal{F}[u] \cdot (\mathbf{1}_{n}\mathcal{F}-\mathcal{F}[u])^\top=\mathcal{F}[u] \cdot (\mathbf{1}_{n}\mathcal{F}-\mathbf{1}_{k})^\top$. This is because every element in $\mathcal{F}[u]$ is either 0 or 1, $\mathcal{F}[u]\cdot \mathcal{F}[u]^\top= \mathcal{F}[u]\cdot \mathbf{1}_{k}^\top$. Thus, this lemma is proved. \noindent\hspace*{\fill}$\blacksquare$
\end{proof}

\begin{lemma} \label{lem:dynamic_cut}
For a community $C$ and any $u \in \bar{C}$,  we have $\mathrm{cut}(C\cup \{u\}, V\setminus (C\cup \{u\}))=\left | \mathrm{cut}(C,\bar{C}) \right |+d(u)-2\left | \mathrm{cut}(u,C) \right | $ and $\mathrm{cut}_a(C\cup \{u\}, V\setminus (C\cup \{u\}))=\left | \mathrm{cut}_a(C,\bar{C}) \right |+d_a(u)-2\left | \mathrm{cut}_a(u,C) \right | $. 
\end{lemma}

\begin{proof}
    By Definition \ref{def:co}, we have $\mathrm{cut}(C\cup \{u\}, V\setminus (C\cup \{u\}))=\left | \mathrm{cut}(C,\bar{C}) \right |-\left | \mathrm{cut}(u,C) \right | +\left | \mathrm{cut}(u,\bar{C}) \right |=\left | \mathrm{cut}(C,\bar{C}) \right |+d(u)-2\left | \mathrm{cut}(u,C) \right |$ due to $\left | \mathrm{cut}(u,C) \right |+\left | \mathrm{cut}(u,\bar{C}) \right |=d(u)$. Similarly, we can prove $\mathrm{cut}_a(C\cup \{u\}, V\setminus (C\cup \{u\}))=\left | \mathrm{cut}_a(C,\bar{C}) \right |+d_a(u)-2\left | \mathrm{cut}_a(u,C) \right | $. Thus, this lemma is proved. \noindent\hspace*{\fill}$\blacksquare$
\end{proof}


\begin{fact} \label{fact:uaec}
Given a community $C$ and any $u \in \bar{C}$, based on Lemma \ref{lem:dynamic_cut}, we have the following facts.   
\begin{equation}\scriptsize\label{eq:utc}
\phi_t(C\cup \{u\})=\frac{ \left | \mathrm{cut}(C,\bar{C}) \right |+d(u)-2\left | \mathrm{cut}(u,C) \right | }{\min \{\mathrm{vol}(C)+d(u),2m-\mathrm{vol}(C)-d(u) \}}.
\end{equation}
\begin{equation} \scriptsize\label{eq:uabc}
\phi_a(C\cup \{u\})=\frac{| \mathrm{cut}_a(C,\bar{C})|+d_a(u)-2| \mathrm{cut}_a(u,C)|}{\min \{\mathrm{vol}_a(C)+d_a(u),\mathrm{vol}_a(V)-\mathrm{vol}_a(C)-d_a(u)\}}.
\end{equation}
\end{fact}

\begin{algorithm}[t]\small
\caption{\textit{Adaptive Community Extractor}}
\label{PCE}
\begin{algorithmic}[noline]
\State \textbf{Input}: An attributed graph $G=(V, E, \mathcal{F})$, a query node $q$, and a parameter $\beta$. 
\State \textbf{Output}: The coarse candidate community $C$.
\end{algorithmic}
    \begin{algorithmic}[1] 
   \State  $d_{a} \leftarrow \mathcal{F} \cdot (\mathbf{1}_{n}\mathcal{F}-\mathbf{1}_{k})^\top$; $\mathrm{vol}_{a}(V) \leftarrow \sum_{v\in V}d_a(v)$
     \State $\mathrm{cut}\gets d(q)$; $\mathrm{vol}\gets d(q)$; $\mathrm{cut}_a\gets d_a(q)$; $\mathrm{vol}_a\gets d_a(q)$
      \State $C \leftarrow \{q\}$; $tmp \leftarrow \{q\}$; $\widehat{\Phi} \leftarrow 1$; $att \gets \mathcal{F}[q]$
            \For {$h=1$ to $\max \{dist(q,u)| u \in G\}$}
            \State $\Delta(h)\leftarrow N^{h}(q)-N^{h-1}(q)$
            \For {each $u\in \Delta(h)$}
            \State $tmp.add(u)$
            \State $\mathrm{cut}\gets \mathrm{cut}+d(u)-2\cdot |C\cap N(u)|$, $vol\gets vol+d(u)$
            \State $\mathrm{cut}_a\gets \mathrm{cut}_a+d_a(u)-2\cdot \mathcal{F}[u]\cdot att^T$, $vol_a\gets vol_a+d_a(u)$
            \State $att\gets att+\mathcal{F}[u]$
            \EndFor
            \If{$\frac{\beta \cdot \mathrm{cut}}{\min \{\mathrm{vol},2m-\mathrm{vol}\}} + \frac{(1-\beta) \cdot \mathrm{cut}_a}{\min \{\mathrm{vol}_a,\mathrm{vol}_a(V)-\mathrm{vol}_a\}}<\widehat{\Phi}$} 
            \State $\widehat{\Phi}\gets \frac{\beta \cdot \mathrm{cut}}{\min \{\mathrm{vol},2m-\mathrm{vol}\}} + \frac{(1-\beta) \cdot \mathrm{cut}_a}{\min \{\mathrm{vol}_a,\mathrm{vol}_a(V)-\mathrm{vol}_a\}}$
            \State $C\gets tmp$
            \EndIf
            \EndFor
            \State \Return $C$
	\end{algorithmic}
\end{algorithm}

By Fact \ref{fact:uaec}, we can know that the challenge with incremental computation of $\phi_{t}$ (or $\phi_a$) is how to efficiently maintain $|\mathrm{cut}(u,C)|$ (or $|\mathrm{cut}_a(u,C)|$). Note that $|\mathrm{cut}(u,C)|=|N(u)\cap C|$ is easily obtained in $O(d(u))$ time. However, calculating \( |\mathrm{cut}_a(u, C)| \) is more complicated. This is due to the fact that $
|\mathrm{cut}_a(u, C)| = |\{(u, v, f) \mid v \in C\}| = \sum_{v \in C} \mathcal{F}[u] \mathcal{F}[v]^\top$ as defined in Definition \ref{def:abc}, which requires \( O(k |C|) \) time to compute. To improve efficiency, we devise several non-trivial data structures to efficiently update  $\phi_{t}$ and  $\phi_a$.


\stitle{Implementation details.}  Algorithm \ref{PCE} provides the pseudo-code of our adaptive community extractor. Specifically, in Lines 1-3, it uses Lemma \ref{lem:att_vol} to obtain the attribute degree vector $d_a$ and initializes the candidate community $C$, current search space $tmp$, and internal attribute $att$. Then, by Lemma \ref{lem:dynamic_cut}, it executes the incremental dynamic update process for $\mathrm{cut}$, $\mathrm{vol}$, $\mathrm{cut}_a$, and $\mathrm{vol}_a$ (Lines 5-10). Lines 11-13 update the candidate community $C$ and the optimal value $\widehat{\Phi}$. Finally, Line 14 returns $C$ as the candidate community.

\begin{example} \label{alg_example}
Reconsider Figure \ref{fig:example}, Figure \ref{fig:algorithm} illustrates the process of extending from the query node $v_0$ to $N^1(v_0)$. To demonstrate the dynamic update algorithm, we focus on the addition of node $v_3$ (i.e., the second step). Initially, the community $C$ consists of nodes $v_0$ and $v_1$. The neighbors of node $v_3$ include $v_0$, $v_1$, $v_2$, and $v_4$, with two edges connecting $v_3$ to community $C$. Consequently, $|\mathrm{cut}(v_3,C)|=2$. Then $\mathrm{cut}$ is updated to $\mathrm{cut}=6+4-2\times 2=6$, while $\mathrm{cut}=8+4=12$. At this juncture, the internal attribute vector $att$ is $[2,1,0,0,0]$, signifying that two edges connecting to the CS attribute and one edge connecting to the ML attribute are associated with community $C$. Therefore, $|\mathrm{cut}_a(v_3,C)|= \mathcal{F}[v_3]\cdot att^T=[1,1,0,0,0] \cdot [2,1,0,0,0]^T=3$, representing the three attribute edges through which node $v_3$ is linked to community $C$ via its own attributes (i.e., CS and ML). The updated $\mathrm{cut}_a$ then becomes $7+5-2\times 3=6$, and $\mathrm{vol}_a=9+5=14$. Following the addition of node $v_3$, $att$ is updated to $[1,1,0,0,0] + [2,1,0,0,0]=[3,2,0,0,0]$. This process continues iteratively, adding nodes sequentially until all nodes (i.e. $\Delta(h)$) are added.
\end{example}


\begin{theorem}\label{thm:pce}
The time complexity and space complexity of Algorithm \ref{PCE} are $O(m+nk)$ and $O(m+n+k)$, respectively. $k$ ($k<<n$) is the attribute dimension.
\end{theorem}

\begin{proof}
    Algorithm \ref{PCE} first takes $O(nk)$ time to  calculate the attribute degree vector $d_a$, $O(n)$ time to obtain $\mathrm{vol}_a(V)$, $\mathrm{cut}$, $\mathrm{vol}$, $\mathrm{cut}_a$, and $\mathrm{vol}_a$  (Lines 1-2). In Lines 4-13, Algorithm \ref{PCE} can incrementally calculate and update attribute-augmented conductance by Fact \ref{fact:uaec}. Specifically, we assume that $u$ is the added vertex in Line 6, it takes $O(d(u))$ time to incrementally update $cut$ and $ O(k)$ time to update $\mathrm{cut}_a$ and $att$. Thus, it takes $O(\sum_{u \in V}(d(u)+k))=O(m+nk)$ to execute Lines 4-13. As a result, Algorithm \ref{PCE} takes $O(m+nk)$  time in total. For the space complexity, Algorithm \ref{PCE} just needs an extra $O(n+k)$ to store $d_a$,  $C$,  $tmp$, and $att$. Besides, it needs $O(m+n)$ to store the input graph $G$. As a consequence,  Algorithm \ref{PCE} needs $O(n+m+k)$  space in total. \noindent\hspace*{\fill}$\blacksquare$

\end{proof}

\begin{figure}[t]
    \centering
\includegraphics[width=0.9\linewidth]{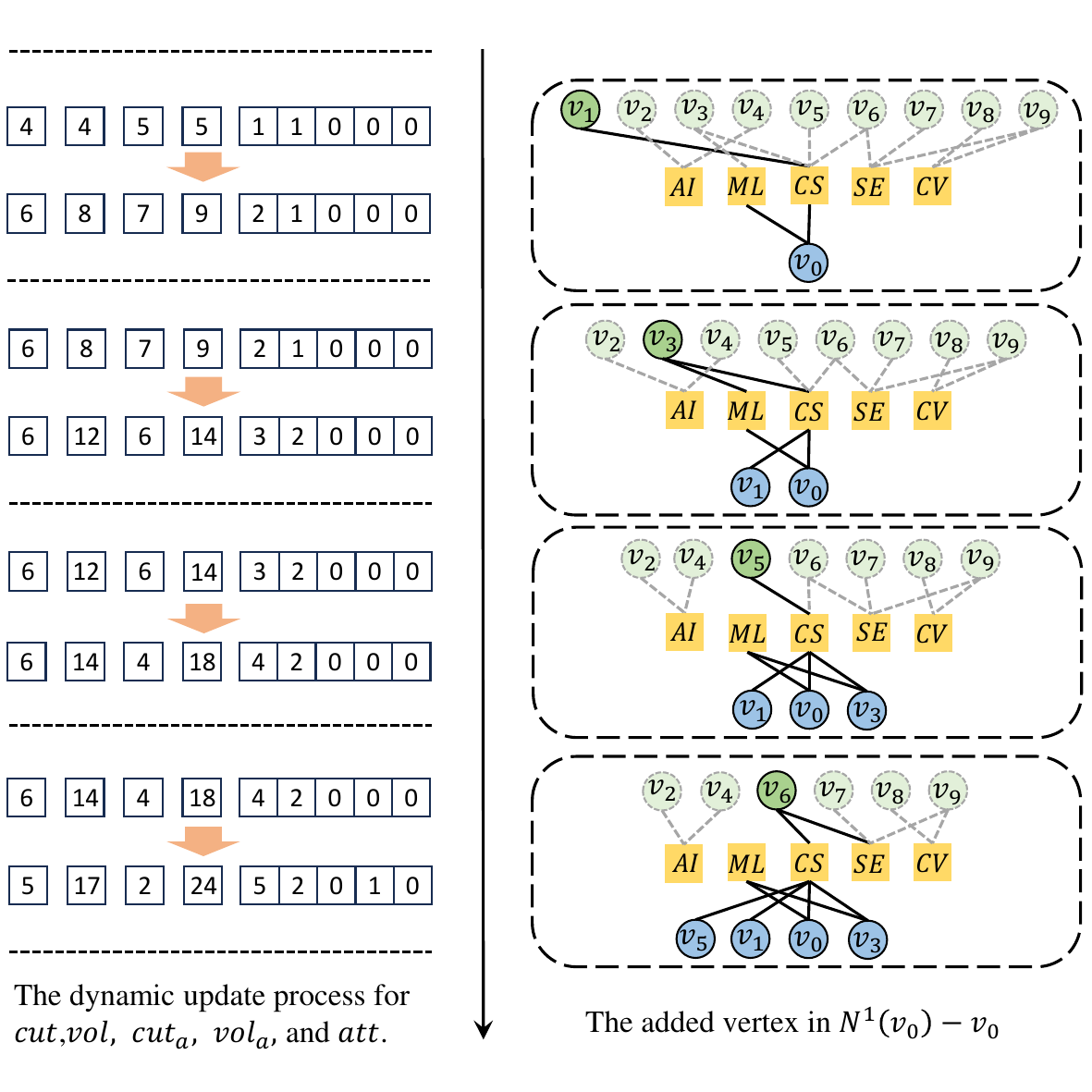}
    \vspace{-5mm}
    \caption{Illustration of Example \ref{alg_example} for Algorithm \ref{PCE}. The left side shows the update process that occurs after each node is added. On the right side, the update of $cut_a$ is elaborated using a node-attribute bipartite graph. The dark blue nodes indicate the current search space $tmp$, while the dark green nodes represent the nodes to be added (best view in color).}
\label{fig:algorithm}\vspace{-0.5cm}
\end{figure} 

\vspace{-0.3cm}
\section{Neural Community Refinement} \label{fcr}
Section \ref{ace} prioritizes scalability over community quality, necessitating further refinement of the coarse candidate community. Therefore, we develop a Neural Community Refinement method, structured into two primary components. In the State Encoder component, we guide the pre-training model using two loss functions that integrate community-aware node representations as state inputs for reinforcement learning. In the Community Refiner component, we employ reinforcement learning to balance exploration and exploitation, thereby identifying nodes to be appended or removed. This process ultimately yields high-quality communities. 
\vspace{-0.3cm}
\subsection{Community-aware State Encoder}  \label{subsec:se}
In this component, our objective is to incorporate community awareness into the embedded space of nodes. To this end, we design two loss functions to achieve a strong correlation between nodes and communities: \textit{Contrastive Loss} and \textit{Triplet Loss}. Intuitively, an edge between two nodes suggests a high probability that they belong to the same community, and the absence of such an edge indicates otherwise. Thus, we model this relationship using contrastive loss:
\begin{equation}\label{CL}\small
L_{C}=\sum_{u,v \in V} A_{uv} \delta (h_{u},h_{v}) + (1-A_{uv}) max(0,\gamma _{1}-\delta (h_{u},h_{v})).
\end{equation}
$\delta(\cdot)$ represents the distance metric, for which we utilize cosine distance.  $\gamma _{1}$ is the margin hyperparameter that indicate the minimum distance. 
In the context of the community search task, when provided with a seed node, the remaining nodes can be classified into two distinct groups: those within the same community and those in different communities. To effectively capture this differentiation, we implement triplet loss:
\begin{equation}\small\label{TL}
L_{T}=\sum_{\left ( q,q_{+},q_{-} \right ) }^{} max\left \{ \delta \left ( h_{q},h_{q_{+}} \right ) -\delta \left ( h_{q},h_{q_{-}} \right ) +\gamma_{2} ,0 \right \}.
\end{equation}
where $q$, $q_+$, and $q_-$ denote the query node, the positive sample (i.e.,  $q_+$ and $q$ are in the same community), and the negative sample (i.e., $q_-$ and $q$ are in different communities), respectively.  $\gamma _{2}$ is the  margin hyperparameter that indicates the minimum distance. 
We concurrently account for both of the aforementioned loss functions. As a consequence, the aggregate loss is formally defined as follows:
\begin{equation}\label{eq:loss}
L=L_{T} + \alpha L_{C}.
\end{equation}
Without loss of generality, in this paper, we utilize  Graph Convolutional Network (GCN \cite{GCN}) as our encoder in this component, defined as follows:
\begin{equation}\small  \label{GCN}
h_{u}^{l+1}=Dr(\sigma (h_{u}^{l}W_{s}^{l+1}+\sum_{v\in N(u)} \frac{h_{v}^{l}W^{l+1}}{\sqrt{(d(u)+1)(d(v)+1)}}+b^{l+1})). 
\end{equation}
Here, $W_{s}^{l+1}$, $W^{l+1}$ and $b^{l+1}$ represents the trainable weights, and $\sigma(\cdot)$ denotes a nonlinear activation function, e.g., ReLU. The parameter $Dr(\cdot)$ refers to the dropout rate, which is used to prevent overfitting. In this paper, we employ a two-layer GCN architecture.

\vspace{-0.2cm}
\subsection{RL-based Community Refiner} \label{subsec:rlcr}
As previously indicated, after acquiring the candidate community, we can redefine the community search problem as a community refinement graph optimization task.

\begin{definition} [Community Refinement Graph Optimization] \label{def:cr}
Given a graph $G$ and a candidate coarse community $C_{coa}$, along with a score function $\Upsilon(.)$, the goal of community refinement is to identify a community $C_{opt}$ with the optimal $\Upsilon(C_{opt})$ by appending or removing nodes for $C_{coa}$.
\end{definition}


\begin{theorem} \label{thm:rl}
The community refinement problem is NP-hard.
\end{theorem}

\begin{proof}
     Let $\Pi=\{<(u_{1},op_{1}),(u_{2},op_{2}),....>| u_{i} \in G, op_{i} \ is \ append \ or \ remove\}$ be the set of all possible refinement orderings, where $(u_{i},append)$ (resp., $(u_{i},remove)$) means that we append (resp., remove) the node $u_i$ for the current community. Thus, by Definition \ref{def:cr}, we can know that the goal of community refinement is to find an ordering $\pi \in \Pi$ such that $\Upsilon(C_{coa}+\pi)$ has the optimal score value. Note that $C_{coa}+\pi$ is a new community by executing the operation $\pi$ for $C_{coa}$. Based on the above concepts, we can reduce the well-known NP-hard problem of the maximum subset sum problem (decision version) to the community refinement problem within polynomial time. Given a set of integers $S=\{s_1,s_2,...s_n\}$ and a target integer $t$, the  subset sum decision problem is to check whether there exists a subset $S'\subseteq S$ such that the sum of the elements in $S'$ equals $t$. From this, we have the following steps. (1) Graph Construction: Construct a graph $G$ where each integer $s_i \in S$ is represented as a node, and let the candidate coarse community $C_{coa}=\emptyset$. (2) Score Function Definition: we define the score function $\Upsilon(C)$ for any community $C$ as follows:
\begin{equation}\small
\Upsilon(C)=
\begin{cases}
    1, & \text{if the sum of the nodes in $C$ is exactly $t$} \\
    0, & \text{otherwise}
\end{cases}
\end{equation}
(3) Polynomial Time Reduction: Given the set $S$ and target $t$, construct the graph $G$, candidate coarse community $C_{coa}$, and the score function $\Upsilon(C)$ as described above. The task of finding a subset $S'\subseteq S$ such that the sum of its elements is $t$ translates to finding a community $C_{opt}\subseteq G$ such that $\Upsilon(C_{opt})=1$. Therefore, this, in turn, is equivalent to finding an ordering $\pi \in \Pi$ such that $\Upsilon(C_{coa}+\pi)=1$, where $C_{coa}+\pi$ represents the community obtained by applying the sequence of operations $\pi$ starting from $C_{coa}$. (4) Equivalence of Problems: If there exists a subset $S'\subseteq S$ such that the sum of its elements is $t$, then there exists an ordering $\pi \in \Pi$ such that $\Upsilon(C_{coa}+\pi)=1$. Conversely, if there exists an ordering $\pi \in \Pi$ such that $\Upsilon(C_{coa}+\pi)=1$, then the corresponding subset of nodes in $G$ (which represents the elements of $S$) has a sum of exactly $t$.
 
Since the subset sum problem (decision version) is NP-hard \cite{IntroductionToAlgorithms2ndEdition}, and we have shown that it can be reduced to the community refinement problem in polynomial time, it follows that the community refinement problem is also NP-hard. \noindent\hspace*{\fill}$\blacksquare$
\end{proof}


There are two primary motivations for employing reinforcement learning (RL for short) in our approach. First, since the community refinement problem is NP-hard (Theorem \ref{thm:rl}), rendering direct solutions impractical. Second, our goal is to achieve adaptive node selection during the refinement stage, eliminating the need to establish a predetermined threshold for nodes or fix the resulting community size (details in Section \ref{subsec:eff}). Fortunately, reinforcement learning naturally has the advantage of solving the above two issues \cite{ppo1,ppo2,tpro}. Thus, we can formalize the community refiner component as a Markov Decision Process (MDP), characterized by the tuple $⟨\mathcal{S},\mathcal{A},\mathcal{P},\mathcal{R},\gamma⟩$ (Section \ref{subsec:rl}), which is  stated as follows:

$\bullet$ \textbf{State.} During the refinement process, we employ the trained state encoder to encode both the community and its neighbors, which subsequently serves as the input for the state representation. We define the initial state $\mathcal{S}_1$ as follows:
\begin{equation}\small \label{eq:state}
\mathcal{S}_{1}(u)=f_s(\mathcal{F}[u], \mathbb{I}\{u\in C_{coa}\} ), \text{where}\ u\in C_{coa}\cup \partial C_{coa}.
\end{equation}
Here, $\partial C_{coa}=\{v \in V \setminus C_{coa}| \exists u \in C_{coa}, (u, v) \in E\}$ is the boundary of $C_{coa}$. $\mathbb{I}$ represents the indicator vector, where it takes the value of 1 when node $u$ is present in the community $C_{coa}$ and 0 otherwise. 

$\bullet$ \textbf{Action.} At each step of this process, we consider the addition of a node from the current state, particularly from the neighbors of the intermediate community, while concurrently removing a node from the community. Ideally, these two operations should operate independently. To achieve this, we design two distinct prediction networks tailored to these different strategies. We utilize multilayer perceptrons (MLP) as the policy networks, which decode the state vector into scalar score values. Specifically, at step $t$, the two policy networks,  $f_{\phi}$ and $f_{\varphi}$, evaluate the action spaces of the two strategies by assigning scores, respectively:
\begin{equation}\label{eq:action}
Score(u)=
\begin{cases}
    f_{\phi}(\mathcal{S}_{t}(u)), & u\in \partial C_t \\
    f_{\varphi}(\mathcal{S}_{t}(u)), & u\in C_t
\end{cases}
\end{equation}
$\partial C_t$ and $C_t$ denote the respective action spaces for node addition and node removal. Separate policy networks are utilized to independently predict the action scores within each of these spaces.

$\bullet$ \textbf{State Transition.} Upon obtaining the scores, we identify the node with the highest score from $\partial C_t$ to be appended into the community, while concurrently removing the node with the lowest score from $C_t$. This procedure yields the new community $C_{t+1}$. We formulate the state transition as follows:
\begin{equation}\label{eq:state_trans}
    \mathcal{S}_{t+1}(u)=f_s(\mathcal{S}_t(u), \mathbb{I}\{u\in C_{t+1}\}), \ u\in C_{t+1}\cup \partial C_{t+1}.
\end{equation}

$\bullet$ \textbf{Reward.} Our main goal is to maximize the score value  of the community (Definition \ref{def:cs}). To achieve this, we define the reward function as the variation in the score value of the target community(i.e., \textit{Obj($\cdot$)}) \cite{10.1145/3514221.3520254,shah2024neurocut,peng2024graphrare,wu2022clear,deng2021unified}. More specifically,
\begin{equation}\label{reward}
\mathcal{R}_{t}=\textit{Obj}\left ( C_{t+1} \right ) -\textit{Obj}\left ( C_{t} \right ). 
\end{equation}
The reward function $\mathcal{R}_{t}$ described above introduces a substantial enhancement to the objective function. This formulation strategically guides the model to prioritize improvements in community value, particularly within regions exhibiting high objective function values, thereby facilitating convergence toward the global optimum. Furthermore, the proposed reward mechanism enhances the flexibility of our framework, rendering it adaptable to any objective function. For networks with available ground-truth labels, established evaluation metrics such as ARI and F1 can be employed. Conversely, for real-world networks without labels, the framework can be guided by user-specified objectives, such as the proposed attribute-augmented conductance, thereby facilitating the fine-tuning of communities to achieve superior structural and attribute-aligned configurations.

Moreover, we propose enhancing the reward mechanism to better address the characteristics of labeled networks:
\begin{equation}\label{r_label}
r_{label}=
\begin{cases}
    1 &\text{if the appending is in}\ \widehat{C} \\
    -1 &\text{otherwise}
\end{cases}
\end{equation}
Here, $\widehat{C}$ represents the ground-truth community. Accordingly, the reward function is revised to $\mathcal{R}_{t}+r_{label}$.

The described above specifically applies to the node addition process. Conversely, during the node removal process, the condition for $r_{label}$ is modified if the removal is not in $\widehat{C}$. This adjustment ensures that the reward mechanism appropriately encourages the elimination of nodes that are not relevant to the community.

\begin{figure*}[t]
    \centering
\includegraphics[width=0.95\linewidth]{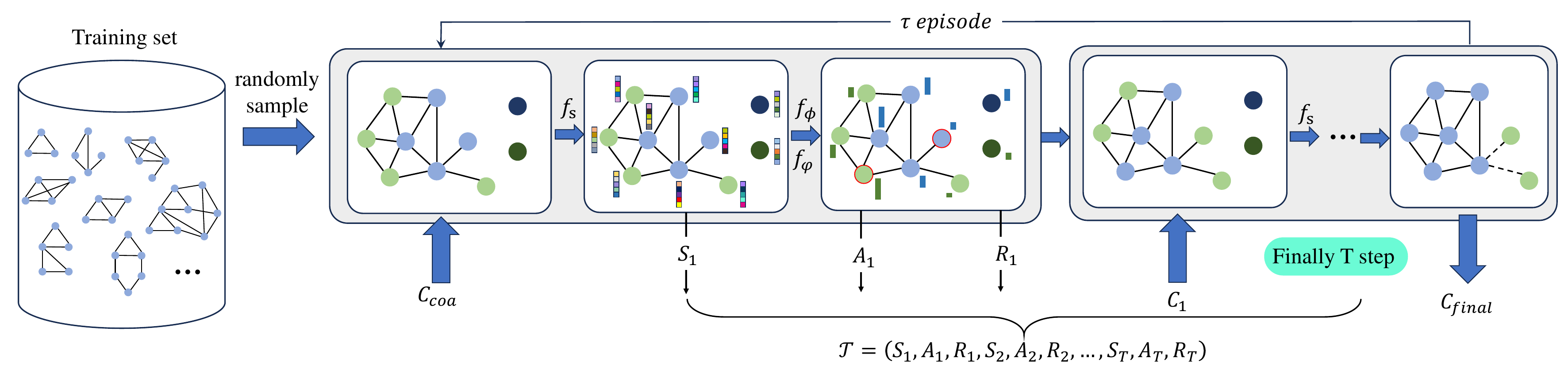}
    \caption{Offline training. Initiating the refinement process from the candidate coarse community, each step generates a new set of states, actions, and rewards, continuing until the predefined termination policy is activated. This procedure is conducted over $\tau$ episodes, during which the agent is trained utilizing the $\tau$ trajectories obtained (best review in color).}
    \label{fig:train}
\end{figure*}


\stitle{Flexible termination strategies.} Existing approaches often rely on a community scoring threshold or limit the number of nodes \cite{ICS-GNN, QD-GNN}, are unsatisfactory due to their restrictive nature. Inspired by these, we introduce two flexible and smooth termination signals, $s_{a}$ and $s_{r}$, to notify the two strategy networks, respectively. Taking the appending process as an example, we conceptualize the termination signal $s_a$ as a virtual node, with its features being randomly generated, and incorporate it into the prediction along with the other node features. The appending process terminates when our RL model selects the $s_{a}$ node. Analogously, the removing process terminates when the model selects the predefined termination node $s_{r}$. Upon termination of both processes, the resulting state is considered the final community. The flexible termination strategy is evaluated in Section \ref{subsec:eff}.

\stitle{Offline training.}  Figure \ref{fig:train} illustrates the details of the offline training. Specifically, during the offline training phase, we initially sample a candidate coarse community from the training dataset and permit the agent to execute the complete community refinement process. Specifically, the agent begins by encoding the state of the community (Eq.~\eqref{eq:state}, illustrated by the color bar in the figure), and subsequently employs the policy network to decode this state into score values (Eq.~\eqref{eq:action}, represented by the blue and green bars), which are then utilized to calculate the reward and update the state. It is crucial to emphasize that, to strike a balance between exploration and exploitation during training, we adopt an $\epsilon$-greedy action strategy. Under this strategy, there is a probability of $1-\epsilon$ that we will select the node with the highest (or lowest) score, while with a probability of $\epsilon$, a random action is taken. The value of $\epsilon$ is linearly annealed from 1.0 to 0.05 throughout the refinement process.  The refinement process is repeated iteratively until the predefined termination strategies are satisfied, resulting in the formation of a trajectory $\mathcal{T}$. During each episode, a single trajectory is used for iterative training, and this procedure continues over $\tau$ episodes until model convergence.





Policy gradient method \cite{pg} is a well-known approach to strategy learning, including Trust Region Policy Optimization (TRPO) \cite{tpro} and Proximal Policy Optimization (PPO) \cite{ppo1,ppo2}. In this paper, we adopt PPO as the key component of our model due to its ability to stabilize the training process by imposing constraints on the policy updates, thus preventing large, destabilizing changes. The loss function for the PPO is delineated below:
\begin{equation}\small\label{PPO}
\mathcal{L}(\phi) = \mathbb{E}_{t} \left[ \min \left( r_t(\phi) A_t, \text{clip}(r_t(\phi), 1 - \varepsilon , 1 + \varepsilon ) A_t \right) \right].
\end{equation} $r_t(\phi)=\frac{f_{\phi_{new}}}{f_{\phi_{old}}}$ represents the ratio of the new policy to the old policy. $A_t$ denotes the advantage function, which is typically estimated using temporal differencing methods. The clipping function $\text{clip}(r_t(\phi), 1 - \varepsilon, 1 + \varepsilon)$ constrains the policy ratio $r_t(\phi)$ within the interval $[1 - \varepsilon,1 + \varepsilon]$. Specifically, if $r_t(\phi)$ is less than $1 - \varepsilon$, it is set to $1 - \varepsilon$; if $r_t(\phi)$ exceeds $1 + \epsilon$, it is set to $1 + \epsilon$. This clipping mechanism is essential to the loss function, as it prevents excessively large updates to the policy during reinforcement learning training, thereby mitigating the risk of instability in the training process.

\stitle{Online refinement.} The online refinement commences with the agent encoding the candidate community and its neighbors into embedding vectors through the state encoder. These vectors are then decoded into node scores using the policy network. Distinct from the training phase, this stage does not involve exploratory actions; instead, the agent directly selects nodes with the highest (or lowest) scores to construct the new community. This iterative process stops until a predetermined termination criterion is met.

\stitle{Remark.} We further elaborate on the advantages of our graph optimization task compared to node classification tasks, such as those represented by TransZero.  When TransZero treats community search as a node classification task, it requires calculating the similarity of all nodes relative to the query node and then assessing changes in community quality, halting the search immediately when the quality declines.  Note that community quality is determined by a function proposed by TransZero, which is limited by the effectiveness of the function itself.  In contrast, we employ a reinforcement learning approach to learn strategies for adding or removing nodes and conditions for stopping the search from training data.  The learned strategies determine which nodes to add or remove and when to stop the search, without any explicit constraints.  Compared to node classification, graph optimization tasks can adapt to different types of networks and flexibly discover high-quality communities.

\begin{table}[t]  
\label{table1}
\caption{Statistics of Datasets.}
\centering 
\resizebox{0.95\linewidth}{!} 
{
\begin{tabular}{c|c|c|c|c|c|c}
\toprule
\label{table:dataset}
\small
\rule{0pt}{10pt} 
\textbf{Dataset}    &$|V|$    &$|E|$    & $|F|$    &$|C|$   &$Mean(C)$    &$Mean(d)$ \\ \midrule 
Cora \cite{TransZero}    &2,708    &10,556    &1,433     &7      &386.86      &3.89    \\  
Facebook \cite{gt}   &3,622    &72,964    &317     &130     &15.62      &40.29    \\  
Cocs  \cite{TransZero}  &18,333    &163,788    &6,806     &15      &1,222.2      &8.93    \\ 
Twitter \cite{gt}   &87,760    &1,293,985    &27,201     &2,838      &10.88      &29.49    \\ \midrule 
Amazon \cite{gt}   &13,178    &3,3767    &658     &4,517      &9.31      &5.12    \\ 
Youtube  \cite{gt}  &216,544    &1,393,206    &2,165     &2,865      &7.67      &12.86    \\ \bottomrule
\end{tabular}}
\vspace{-5mm}
\end{table}
\section{Experimental Evaluation}\label{sec:experiments}

In this section, we conduct comprehensive experiments to evaluate the effectiveness and efficiency of our solutions. All experiments are conducted on a Linux machine with an Intel Xeon(R) Silver 4210 @2.20GHz CPU and 1 TB RAM. 

\subsection{Experimental Setup}

\noindent \textbf{Datasets.} 
Our solutions have been rigorously evaluated on six publicly available datasets with ground-truth communities (Table \ref{table:dataset}), which serve as well-established benchmarks for the community search problem \cite{huang2014querying,k-core2,yao2021efficient,k-clique1,DMCS,k-ecc,ICS-GNN,QD-GNN,CommunityAf,coclep,cgnp,TransZero}. In these datasets, $|F|$ represents the number of attributes, $|C|$ denotes the number of communities, $Mean(C)$ indicates the average size of the communities, and $Mean(d)$ reflects the average degree of the nodes. We classify the ground-truth communities into three categories: training, validation, and testing communities. The training communities are utilized for model training, the validation communities for fine-tuning parameters, and the testing communities for assessing the model's performance. Consistent with previous research \cite{ICS-GNN,QD-GNN,CommunityAf,coclep,cgnp,TransZero}, we also randomly allocate these categories in  5:1:4.



\noindent \textbf{Algorithms.} 
The following ten baselines are compared. (1) Traditional rule-based community search methods:  CTC \cite{k-truss3} based on  $k$-truss, CST \cite{k-core2} based on $k$-core,  M$k$ECS \cite{k-ecc} based on $k$ edge connected component, and conductance \cite{DBLP:conf/focs/AndersenCL06,DBLP:conf/kdd/KlosterG14,DBLP:conf/aaai/LinLJ23}. Note that we  do not include density \cite{DBLP:journals/pvldb/WuJLZ15,DBLP:conf/sigmod/DaiQC22,DBLP:conf/kdd/YeLLLLW24} and modularity \cite{DMCS} in the experiments because they ignore the external sparsity and are outperformed by conductance \cite{DBLP:journals/pacmmod/KamalB24,DBLP:journals/eswa/HeLYLJW24}. (2) Learning-based methods:  ICS-GNN \cite{ICS-GNN}, QD-GNN \cite{QD-GNN},  COCLEP \cite{coclep} CGNP \cite{cgnp}, CommunityAF \cite{CommunityAf} TransZero \cite{TransZero}. NSCAC is our proposed algorithm.  There is no comparison with \emph{ALICE} \cite{ALICE}, as the authors did not provide the source code, and reproducing the results is highly challenging. However, we discuss its drawbacks in Sections \ref{sec:intro} and \ref{sec:existing}.

\noindent \textbf{Metrics.} Several well-established metrics are used to evaluate the quality of the identified communities \cite{coclep,TransZero,DMCS}: the F1-score \cite{f1}, Normalized Mutual Information (NMI) \cite{nmi}, and Jaccard similarity (JAC) \cite{jac}. Higher values for the F1-score, NMI, and JAC signify better performance. We also report the runtime to test the efficiency.  Unless otherwise specified, for enhanced reliability, we select one vertex at random from each ground-truth community to serve as the query vertex and  report the average runtime and quality.

\noindent \textbf{Parameters.} Unless specified otherwise, we take the default parameters of the ten baselines. On the other hand, our model NCSAC involves three main parameters:  $\beta$ (the preference parameter in Eq. \ref{eq:aec}), $\alpha$ (the loss function parameter in Eq. \ref{eq:loss}), and  $\tau$ (the number of episodes in offline training of Section \ref{subsec:rlcr}). We set $\beta=0.2$, $\alpha=0.5$, $\tau=2000$ by default. 

\subsection{Effectiveness Evaluation} \label{subsec:eff}
 \stitle{Exp-1: Comparison of our algorithm NCSAC with baselines.} Figure \ref{fig:effectiveness}  provides a detailed box plot showing the F1-score, NMI, and JAC values achieved by different algorithms across all datasets. Besides, Table \ref{table:metric} also shows the median scores of different methods for a clearer comparison. Note that we only explain the F1-score because NMI and JAC  have similar results. For F1-score, we have the following observations: (1)  Our model \ourname \ achieves the best median scores on five of the six datasets (on Cocs, \ourname \ is slightly worse than TransZero), with F1-score improvement of 5.3\%$\sim$42.4\%. For example, on Facebook, the F1-score of \ourname \ is 0.47, whereas the runner-up is 0.33, representing an improvement of 42.4\% in the F1-score. This is because our model integrates the strengths of both rule-based and learning-based methodologies.  However, baseline methods study rule-based approaches and learning-based approaches independently,  resulting in suboptimal community quality. (2) Our proposed \ourname\ consistently outperforms all baselines in most cases, except for being slightly less effective than QD-GNN and COCLEP on Facebook, in terms of the quality lower bound. This outcome is attributed to \ourname's robust two-stage design. Specifically, baseline methods typically add nodes sequentially from a query node based on predefined strategies, which can be susceptible to noise, resulting in early termination and suboptimal community quality. Essentially, some nodes that could potentially improve the community score are excluded due to the limitations of these models. In contrast, \ourname\ initially utilizes a novel attribute-augmented conductance method to identify a satisfactory and interpretable candidate coarse community. This community is then refined using reinforcement learning, ensuring a more consistent and superior performance. (3) \ourname\ exhibits stability that is on par with (indicated by a narrower box width) other learning-based methods. Besides, learning-based models are consistently better than traditional rule-based models (e.g., CTC, CST, and MKECS)  across all datasets. This can be attributed to the fact that learning-based models possess the flexibility to capture both structural proximity and attribute similarity within communities, a feat that traditional rule-based models are unable to achieve.  In short, these results give clear evidence that our model can find higher-quality communities when contrasted with baselines.

\begin{figure*}[t]
    \centering
    \includegraphics[width=0.95\linewidth]{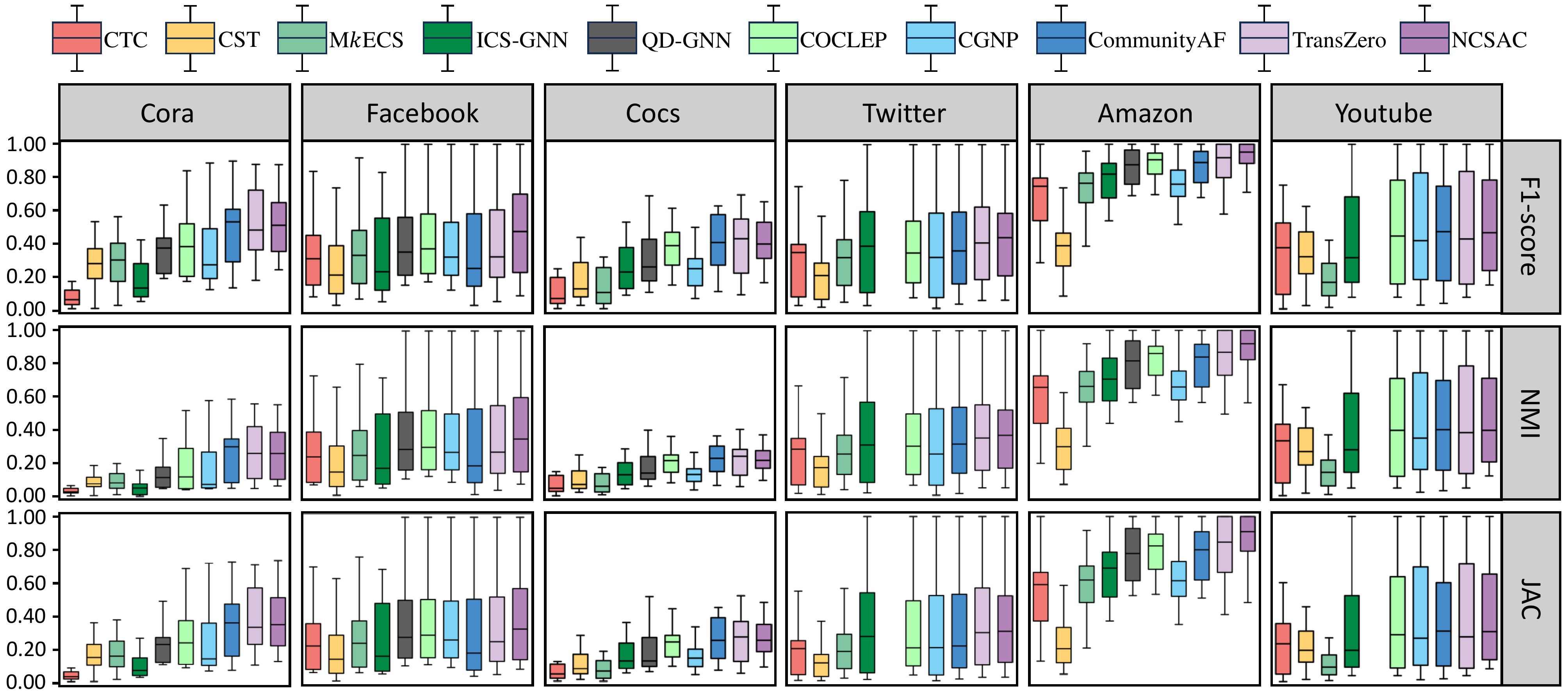}
    \caption{Comparison of our algorithm NCSAC with baselines. Since  QD-GNN runs out of memory on Twitter and YouTube, we only report their results on the four remaining datasets (best view in color).}
    \label{fig:effectiveness}
\end{figure*}

  \begin{table*}[t!]
	\caption{\small Median scores of different methods. OOM indicates that the corresponding methods ran out of memory.} 
	\centering
	\footnotesize
	\scalebox{1}{
		\begin{tabular}{c|ccccccccc}
			\toprule
			\multicolumn{1}{c|}	{F1/NMI/JAC} & Cora& Facebook&Cocs& Twitter&Amazon& Youtube\\
			\midrule
			
			\multirow{8}{*} \textit{CTC} & 0.05/0.03/0.02 & 0.30/0.22/0.23 &0.06/0.04/0.03  & 0.34/0.20/0.21 &0.74/0.58/0.65& 0.38/ 0.34/0.23\\
			
			\textit{CST}&0.27/0.14/0.07 & 0.21/0.14/0.14 & 0.14/0.07/0.09  & 0.20/0.11/0.12 & 0.38/0.20/0.29& 0.31/0.26/0.20\\
			\textit{MkECS}  & 0.29/0.15/0.08 &0.32/0.23/0.24 & 0.13/0.05/0.07 & 0.31/0.18/0.19 & 0.76/0.61/0.66& 0.18/0.14/0.10\\
			\textit{ICS-GNN}& 0.12/0.06/0.04 & 0.22/0.15/0.16& 0.22/0.12/0.13 & 0.38/0.27/0.28 & 0.81/0.68/0.70& 0.31/0.27/0.20&\\
			\textit{QD-GNN}& 0.36/0.22/0.11 & 0.34/0.27/0.28& 0.25/0.13/0.13 & OOM & 0.87/0.77/0.81& OOM \\	
			\textit{COCLEP}&  0.37/0.23/0.11 & 0.36/0.28/0.29& 0.39/0.21/0.21  & 0.34/0.20/0 & 0.90/0.82/0.86&0.44/0.40/0.30\\
			\textit{CGNP}&  0.26/0.14/0.07 & 0.31/0.25/0.26& 0.25/0.14/0.16  & 0.31/0.20/0.21 & 0.75/0.61/0.65 &0.41/0.37/0.27\\
			\textit{CommunityAF}&\textbf{0.52}/\textbf{0.35}/0.28 &0.25/0.17/0.18   &0.40/0.22/0.25 &0.37/0.31/0.23 & 0.88/0.80/0.83 &0.44/0.40/0.31\\
			\textit{TransZero}&0.47/0.33/0.26 &0.32/0.24/0.26& \textbf{0.43}/\textbf{0.24}/\textbf{0.29} & 0.40/0.34/0.30& 0.91/0.84/0.86&0.42/ 0.38/0.29\\
			
			\textit{NCSAC}  &0.50/0.34/\textbf{0.29} & \textbf{0.47}/\textbf{0.32}/\textbf{0.34}& 0.42/0.23/0.27 &\textbf{0.41}/\textbf{0.35}/\textbf{0.31} & \textbf{0.94}/\textbf{0.89}/\textbf{0.91} & \textbf{0.46}/\textbf{0.41}/\textbf{0.33}\\
			\bottomrule	
	\end{tabular}} 
	\label{table:metric}
\end{table*}


 \begin{table}[t]
\renewcommand{\arraystretch}{1.2}
\caption{Comparison of the traditional conductance  and the attribute-augmented conductance.}\vspace{-0.3cm}
\resizebox{0.95\linewidth}{!} 
{
\begin{tabular}{c|ccc|ccc}
\hline
\label{table:con}
\small
\rule{0pt}{10pt} 
\multirow{2}{*}{}   &\multicolumn{3}{c|}{Traditional Conductance}  &\multicolumn{3}{c}{Attribute-augmented Conductance} \\
\cline{2-7}
        &F1-score    &NMI    &JAC    &F1-score    &NMI    &JAC \\ \hline
Cora    &0.2195    &0.0712    &0.1348     &0.3721      &0.1281      &0.2486    \\
Facebook    &0.3059    &0.2403    &0.2191     &0.4145     &0.3171      &0.2938    \\
Cocs    &0.2562    &0.0858    &0.1522     &0.3207      &0.1580      &0.2121    \\
Twitter    &0.2292    &0.1968    &0.1428     &0.3647      &0.3033      &0.2469    \\
Amazon    &0.8666    &0.8139    &0.7752     &0.8892      &0.8381      &0.8038    \\
Youtube    &0.2664    &0.2380    &0.1988     &0.3995      &0.3511      &0.2882    \\ \hline
Average $+/-$ &-     &-       &-       &+0.1028       &+0.0749       &+0.0784 \\ \hline 
\end{tabular}}\vspace{-0.5cm}
\end{table}

 
 \stitle{Exp-2: Comparison of the traditional conductance and the attribute-augmented conductance.} Table \ref{table:con} offers a comparative analysis between traditional conductance and the newly introduced attribute-augmented conductance.  To isolate the evaluation to the community extraction process alone, we report the quality of the candidate coarse community without considering the subsequent refinement stage. As can be seen, the results demonstrate that our proposed attribute-augmented conductance consistently identifies more accurate communities compared to traditional conductance. Remarkably, attribute-augmented conductance achieves an average improvement of 0.1028 in F1-score, 0.0749 in NMI, and 0.0784 in JAC across all datasets, highlighting its exceptional ability to discern community structures.

\begin{figure*}[t!]
    \centering
    \subfigure[F1-score with varying $\beta$]{
    \includegraphics[height=3cm]{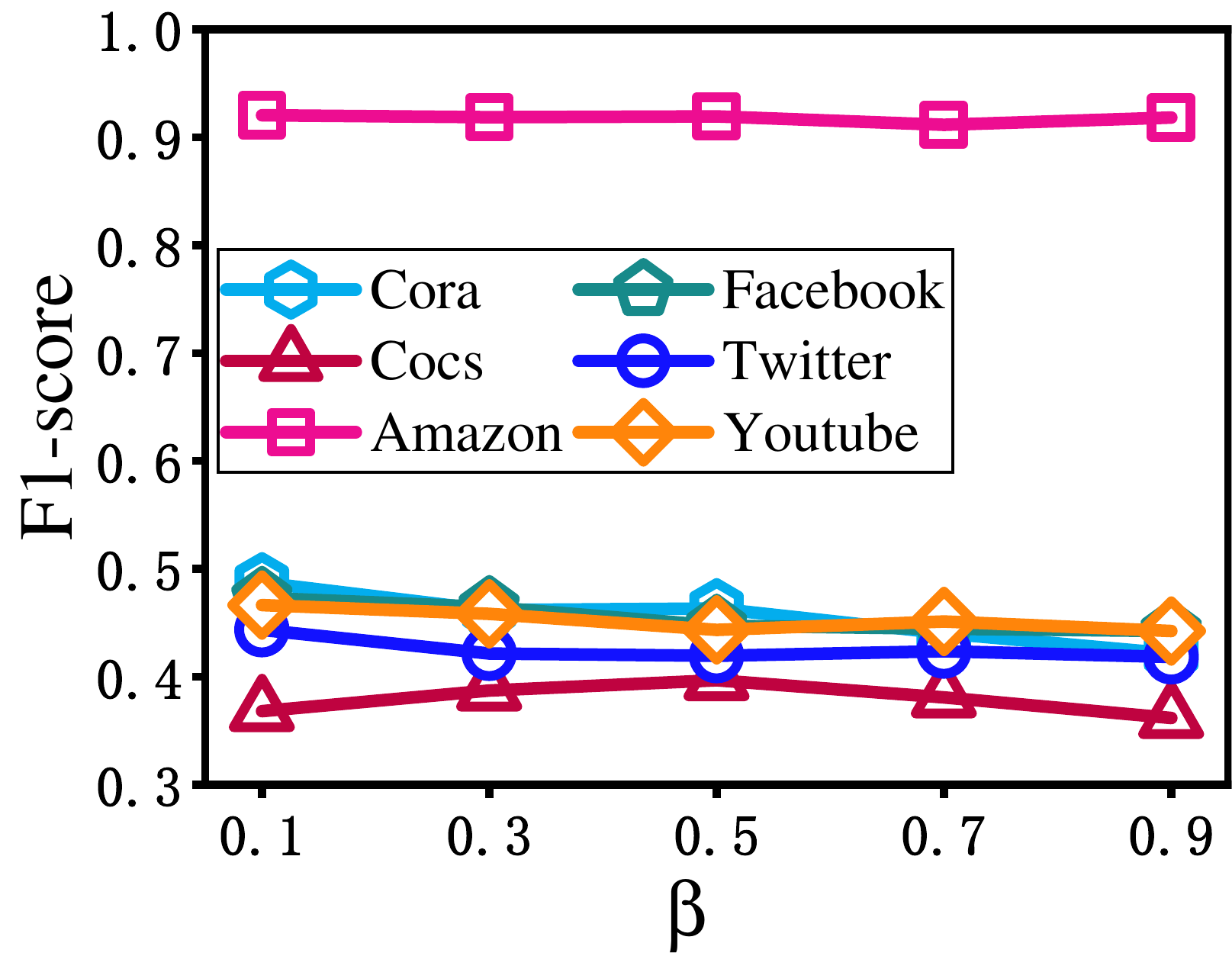}
    }
    \subfigure[F1-score with varying $\alpha$]{
    \includegraphics[height=3cm]{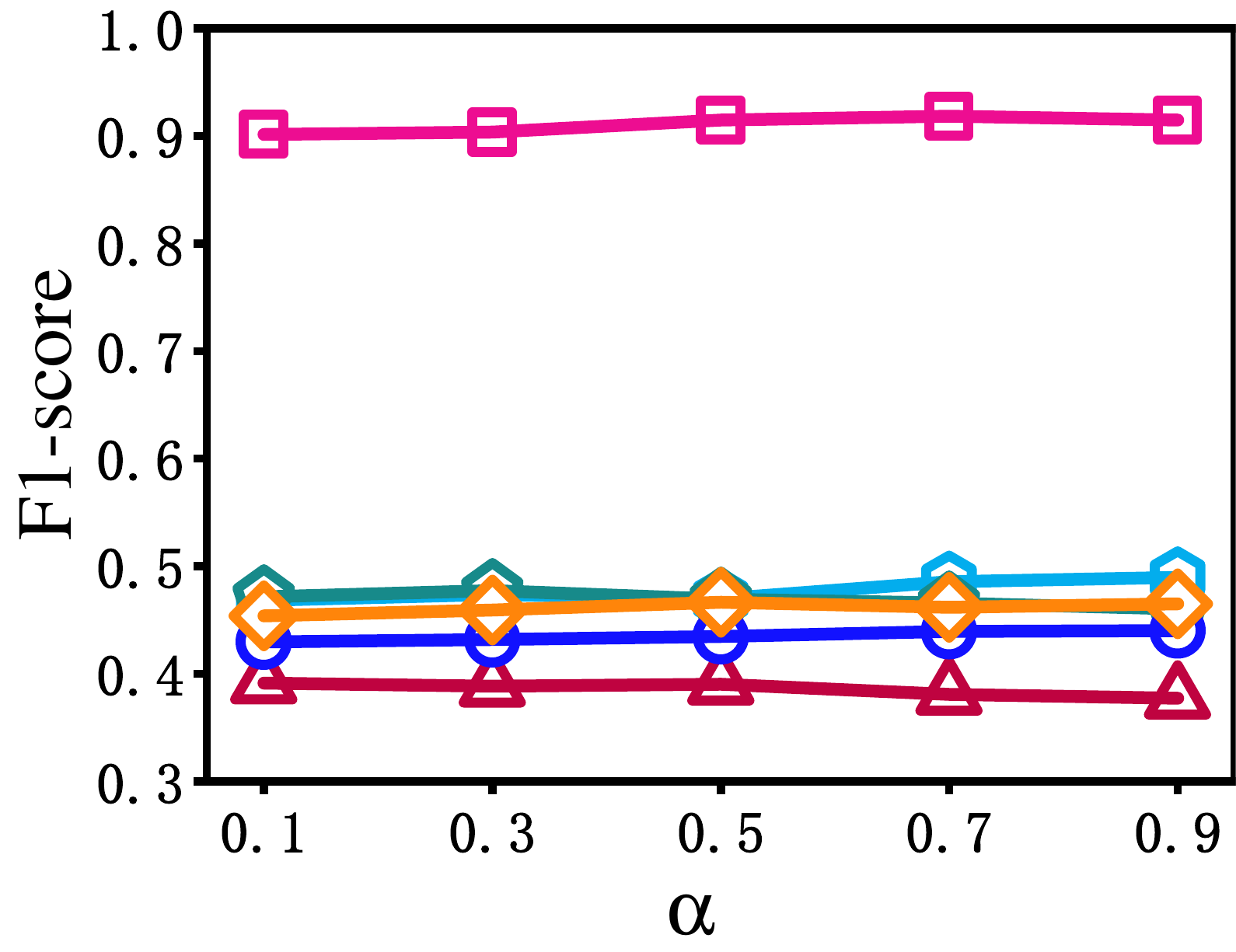}
    }
    \subfigure[F1-score with varying $\tau $]{
    \includegraphics[height=3cm]{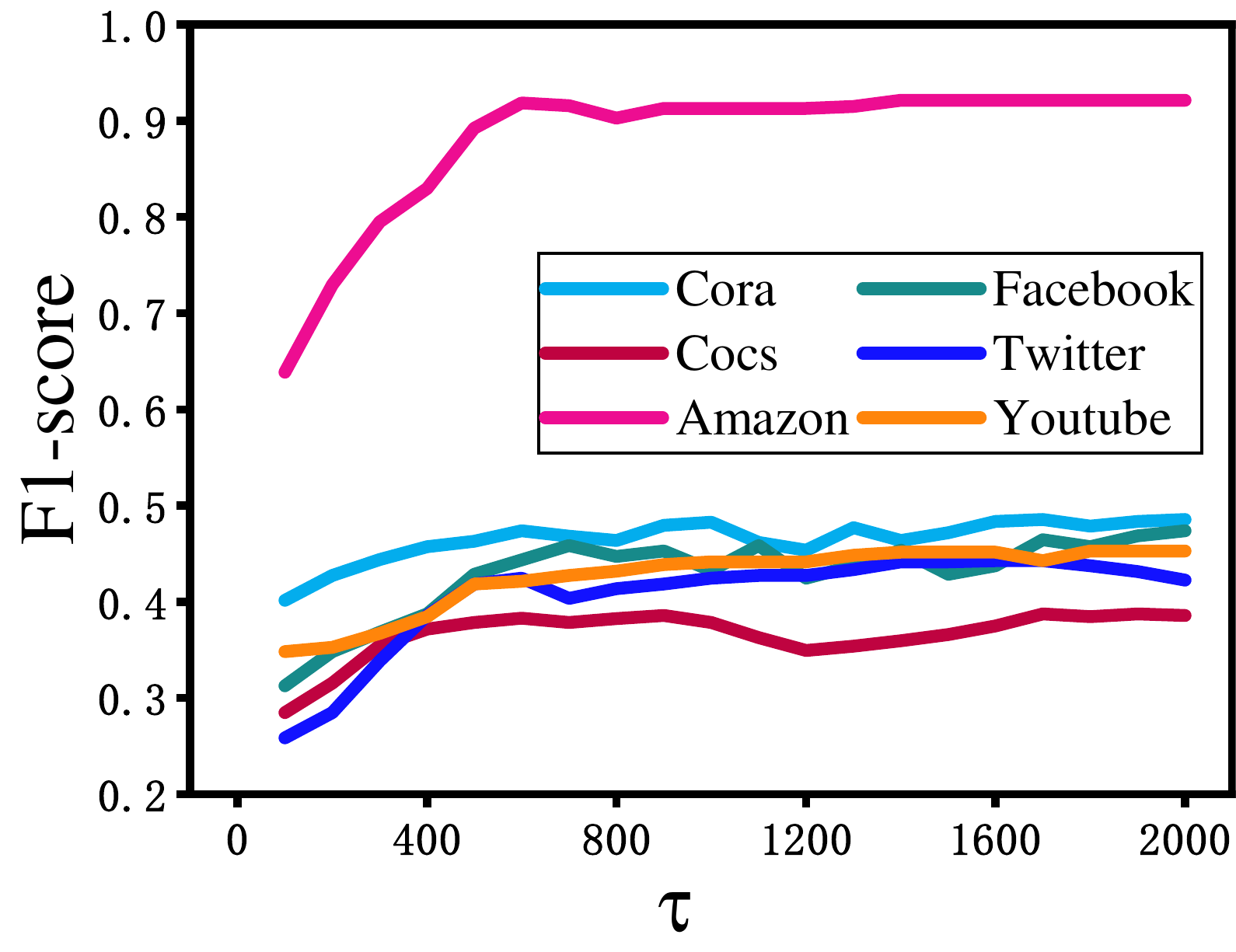}
    }
     \subfigure[Relative F1-score with varying $\eta $]{
    \includegraphics[height=3cm]{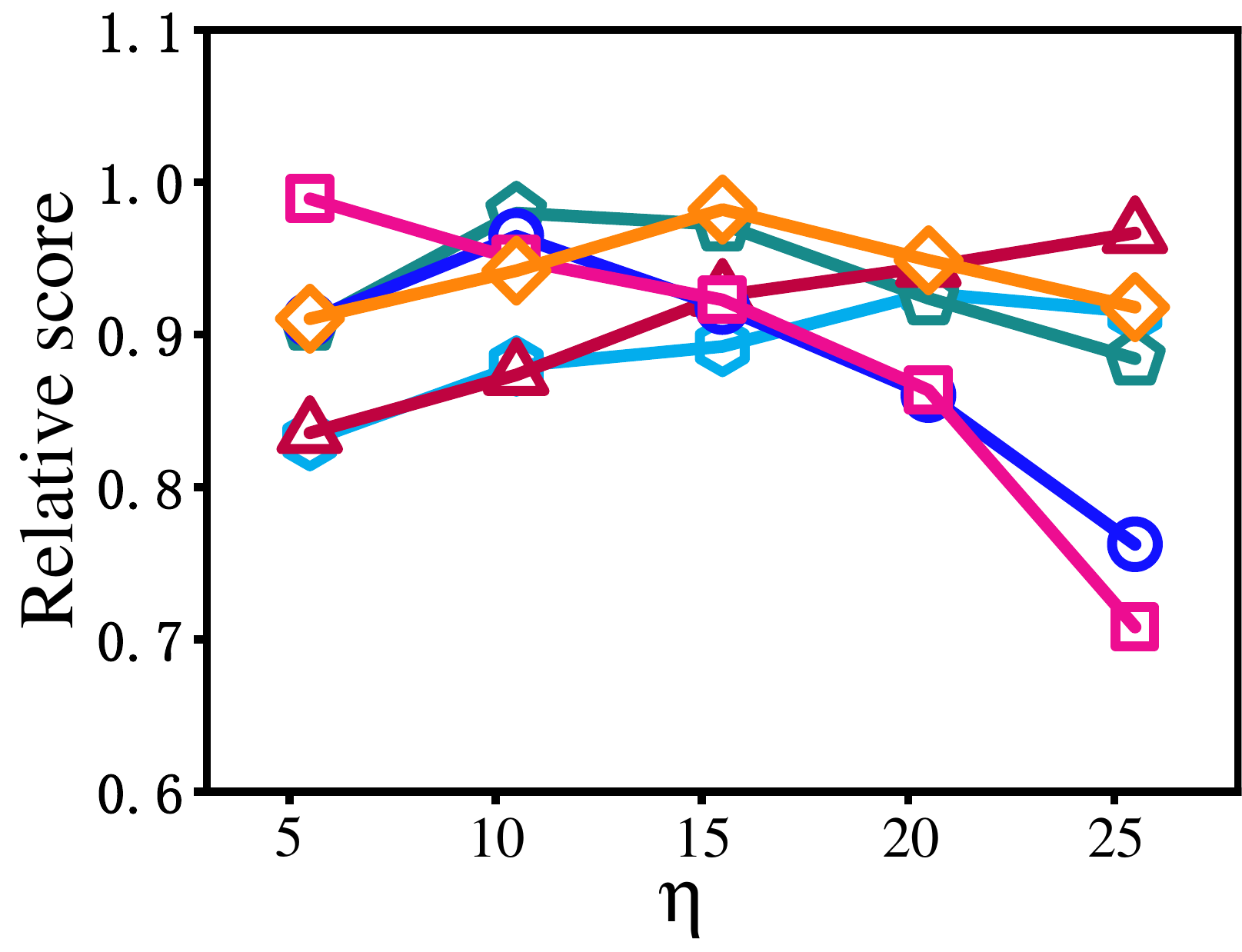}
    }
    \vspace{-0.2cm}
    \caption{Performance of NCSAC with varying parameters.} \vspace{-0.3cm}
    \label{fig:hyperparameters}
\end{figure*}

\stitle{Exp-3: Performance of NCSAC with varying parameters.} This experiment investigates how the parameters  affect the performance of our \ourname. 
Figure \ref{fig:hyperparameters}  only reports the F1-score,
with analogous trends observed on NMI and JAC. Specifically, we have the following conclusions: 
(1) Figure \ref{fig:hyperparameters} (a) presents the results obtained by varying $\beta$, which strikes a balance between topology-based conductance and attribute-based conductance in Equation \ref{eq:aec}. As depicted, on five out of six datasets (with Amazon being an exception showing no significant change with $\beta$), the performance of \ourname \ diminishes as $\beta$ rises. This trend can be intuitively understood as follows: as $\beta$ increases, the weight assigned to topological factors grows, thereby overshadowing the influence of attribute-based factors and subsequently resulting in a decline in quality. This finding highlights the pivotal role of attribute similarity and validates the effectiveness of our proposed attribute-augmented conductance. (2) Figure \ref{fig:hyperparameters} (b) illustrates the performance of \ourname \ with varying  $\alpha$, which serves to balance the triplet loss and contrastive loss in Equation \ref{eq:loss}. As can be seen,  The F1-score remains relatively stable as $\alpha$ increases, indicating that the impact of the triplet loss significantly exceeds that of the contrastive loss. Subsequent ablation studies provide further support for this observation.  (3) Figure \ref{fig:hyperparameters} (c) presents the influence of the number of episodes $\tau$ on the F1-score. 
As observed, the F1-score experiences a rapid increase (by 0.2-0.3) within the initial 800 episodes, followed by a gradual slowdown (of less than 0.05) in the subsequent 800-2000 episodes. This finding suggests that a relatively small number of episodes are sufficient to achieve a high F1-score. 
(4) Let $\eta$ be the fixed number of nodes refined during the refinement phase (section \ref{subsec:rlcr}) and $RFS=\frac{F1-score \ obtained \ by \ FixNCSAC}{F1-score \ obtained \ by \ NCSAC}$ be  the relative F1-score, where FixNCSAC is the NCSAC with the fixed number of nodes refined. This experiment studies the effectiveness of the proposed flexible termination strategy (details in Section \ref{subsec:rlcr}). Figure \ref{fig:hyperparameters} (d) shows the  results. As illustrated, \emph{RFS} consistently remains below 1 and exhibits no discernible pattern of variation across all datasets. Thus, this result provides compelling evidence that our proposed flexible termination strategy consistently outperforms the fixed-node refinement strategy.





\subsection{Efficiency Evaluation} \label{subsec:runtime}
\stitle{Exp-4: Runtime of different algorithms.} Figure \ref{fig:efficiency} depicts the runtime  of various algorithms across all datasets. The observations are as follows: (1) Traditional rule-based methods (namely CTC, CST, and MkECS) exhibit consistently faster execution speeds compared to learning-based approaches, albeit at the expense of significantly compromised quality, as discussed in Section \ref{subsec:eff}. This is attributable to their reliance on predefined structural constraints for localized community search, eliminating the need for intricate learning processes. However, the inflexibility of these patterns undermines the quality of the results. (2) Our proposed model, \ourname, demonstrates runtime efficiency comparable to other learning-based methods. Notably, \ourname \ substantially surpasses alternative approaches on datasets with well-defined structures, such as Amazon and Twitter. This superiority stems from the efficacy of our community extractor, which swiftly identifies candidate community structures, thereby minimizing the necessity for extensive refinement. Conversely, in datasets with less discernible community structures, like Cora and Cocs, TransZero shows faster query times due to its simpler termination condition, which ceases the search upon the initial decline in community score. In contrast, CommunityAF employs a more elaborate termination criterion based on a sliding window approach, leading to longer query time. (3) ICS-GNN emerges as the most time-intensive method, as it necessitates retraining for each individual query, equating its query time to the training time. Consequently, ICS-GNN incurs the highest computational cost among all evaluated methods.
In summary, these results give some preliminary evidence that our model is indeed efficient in practice.

\begin{figure}[t!]
    \centering
\includegraphics[width=0.9\linewidth]{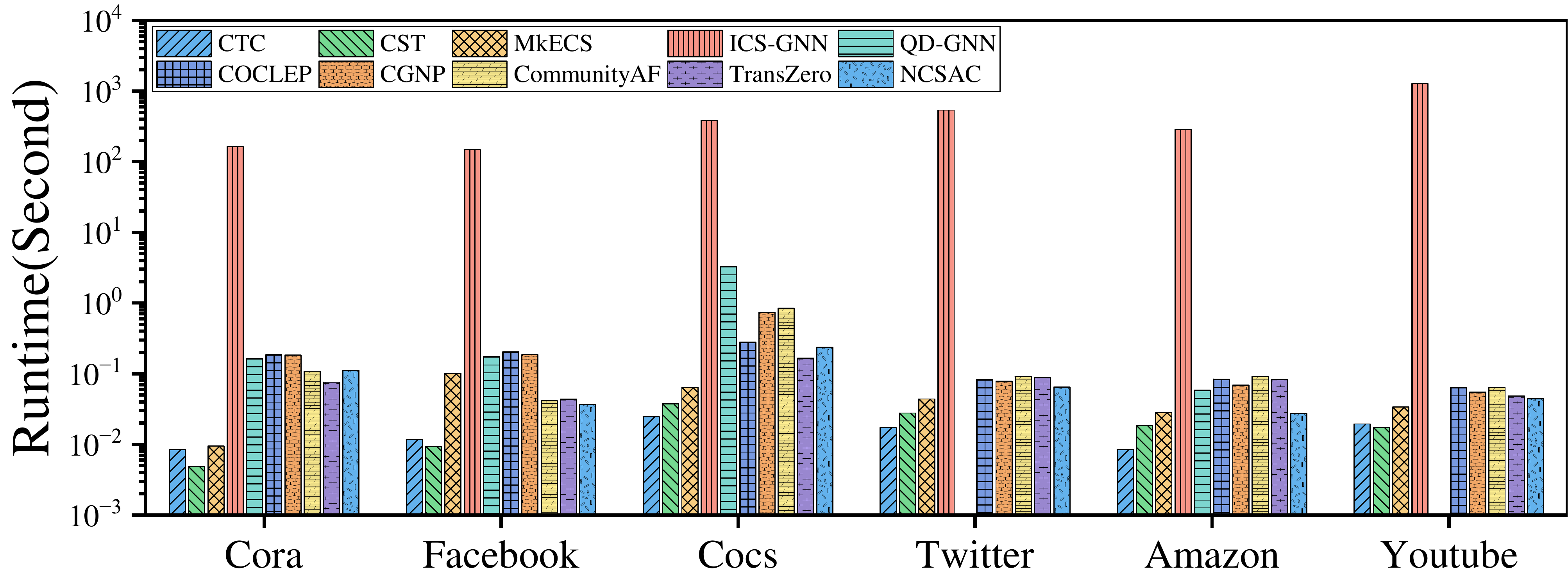}\vspace{-0.3cm}
\caption{Runtime of different algorithms (best view in color).}
\label{fig:efficiency}\vspace{-0.3cm}
\end{figure}

\begin{table}[t]
\caption{Efficiency of adaptive community extractor.}\vspace{-0.3cm}
\centering 
\resizebox{0.95\linewidth}{!} 
{
\begin{tabular}{c|cccccc}
\toprule
\label{table:compare}
\rule{0pt}{10pt} 
\textbf{Dataset}    &Cora    &Facebook    &Cocs    &Twitter   &Amazon    &Youtube \\ \midrule 
\emph{NAC}  &0.3128    &0.1174    &3.8346    &24.8306      &0.0231      &4.1999    \\ 
\emph{ACE}  &0.02    &0.0097   &0.1153     &0.0131      &0.0003      &0.0338    \\ \bottomrule
\end{tabular}}\label{tab:alg}\vspace{-0.3cm}
\end{table}

\stitle{Exp-5: Efficiency of adaptive community extractor.} To evaluate the efficiency of our adaptive community extractor (i.e.,  \emph{ACE} in Algorithm \ref{PCE}), we implement the naive method, named \emph{NAC} (discussed as a  failed attempt in Section \ref{ace}). Namely, \emph{NAC}  materializes the dense multigraph (Figure \ref{fig:example}) for computing the attribute-augmented conductance. As shown in Table \ref{tab:alg}, our proposed \emph{ACE} is at least 12$\times $ faster than the naive method \emph{NAC}. Specifically,  \emph{ACE} is 15$\times $, 12$\times $, 33$\times $, 1895$\times $, 77$\times $, 124$\times $ faster than \emph{NAC} on Cora, Facebook, Cocs, Twitter, Amazon, and Youtube respectively. 
This is because  
our \emph{ACE} requires only linear time complexity (Theorem \ref{thm:pce}), whereas the worst-case scenario for \emph{NAC} entails a quadratic time complexity. Thus, these results align with our theoretical analysis (Section \ref{ace}). 

 \begin{figure}[t!]
    \centering
\includegraphics[width=0.95\linewidth]{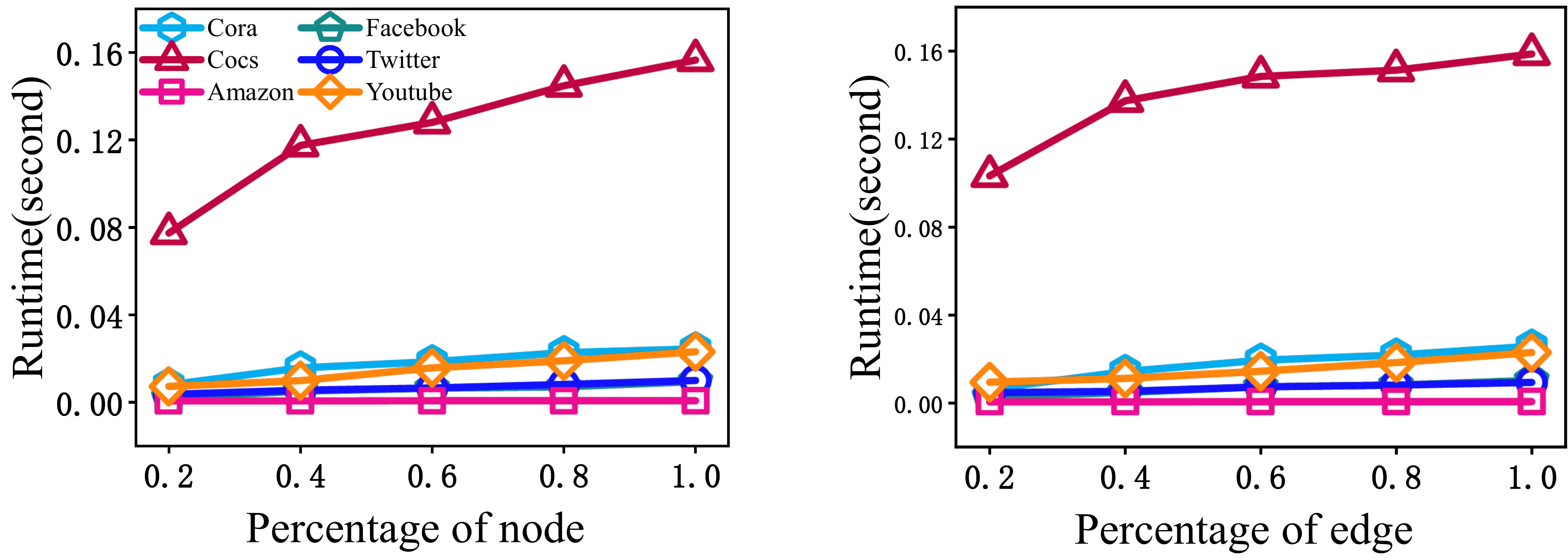}
\caption{Scalability testing.}
    \label{fig:scab}
\vspace{-0.3cm}
\end{figure}

\stitle{Exp-6: Scalability testing.} 
To further evaluate the scalability of our adaptive community extractor, we first generate 60 synthetic subgraphs by randomly sampling 20\%, 40\%, 60\%, 80\%, and 100\% of vertices or edges from each of the six datasets. Subsequently, we present the runtime on these subgraphs in Figure \ref{fig:scab}.  As depicted, the runtime of Algorithm \ref{PCE} scales nearly linearly with the size of the subgraphs. These findings suggest that our proposed algorithm can deal with massive graphs.

\subsection{Ablation Studies and Case Studies}

\begin{figure*}[t!]
    \centering
    \subfigure{
    \includegraphics[width=0.3\linewidth]{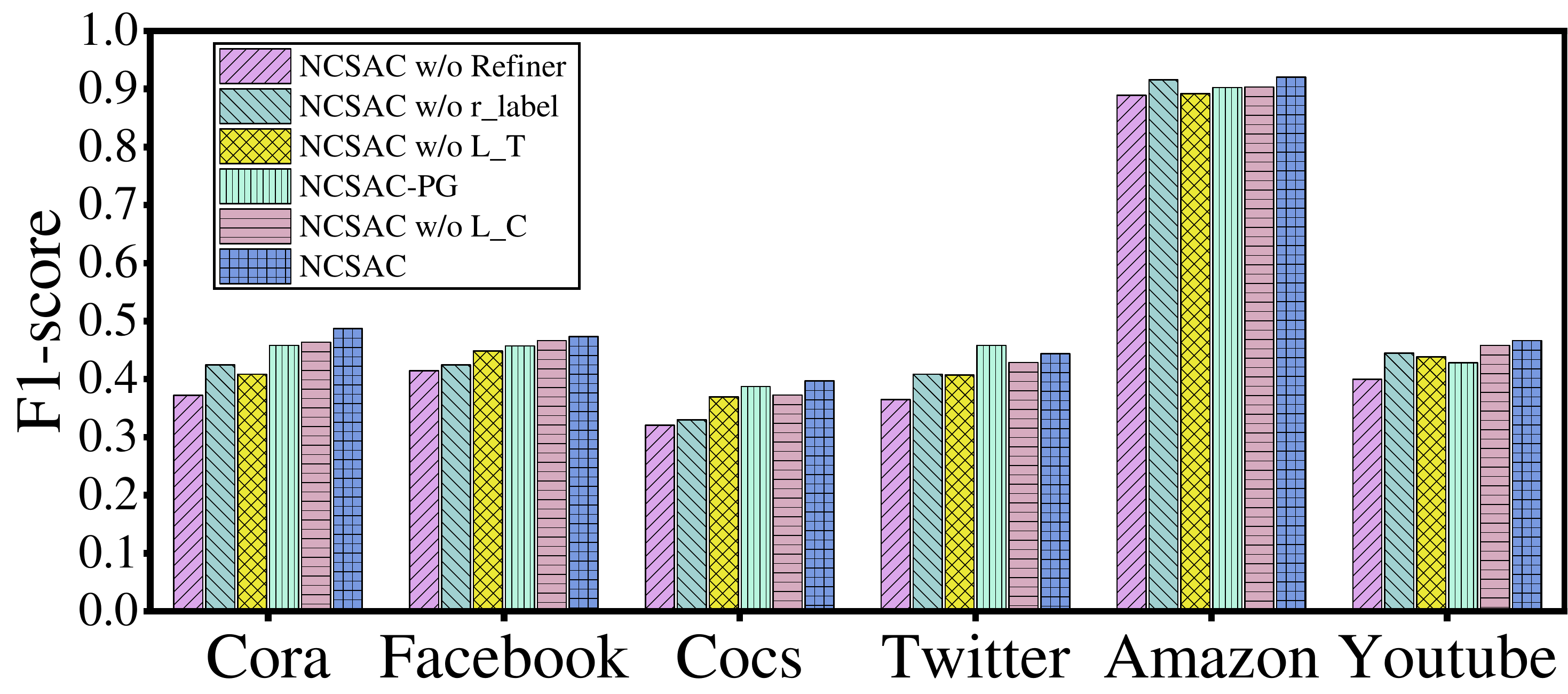}
    }
    \subfigure{
    \includegraphics[width=0.3\linewidth]{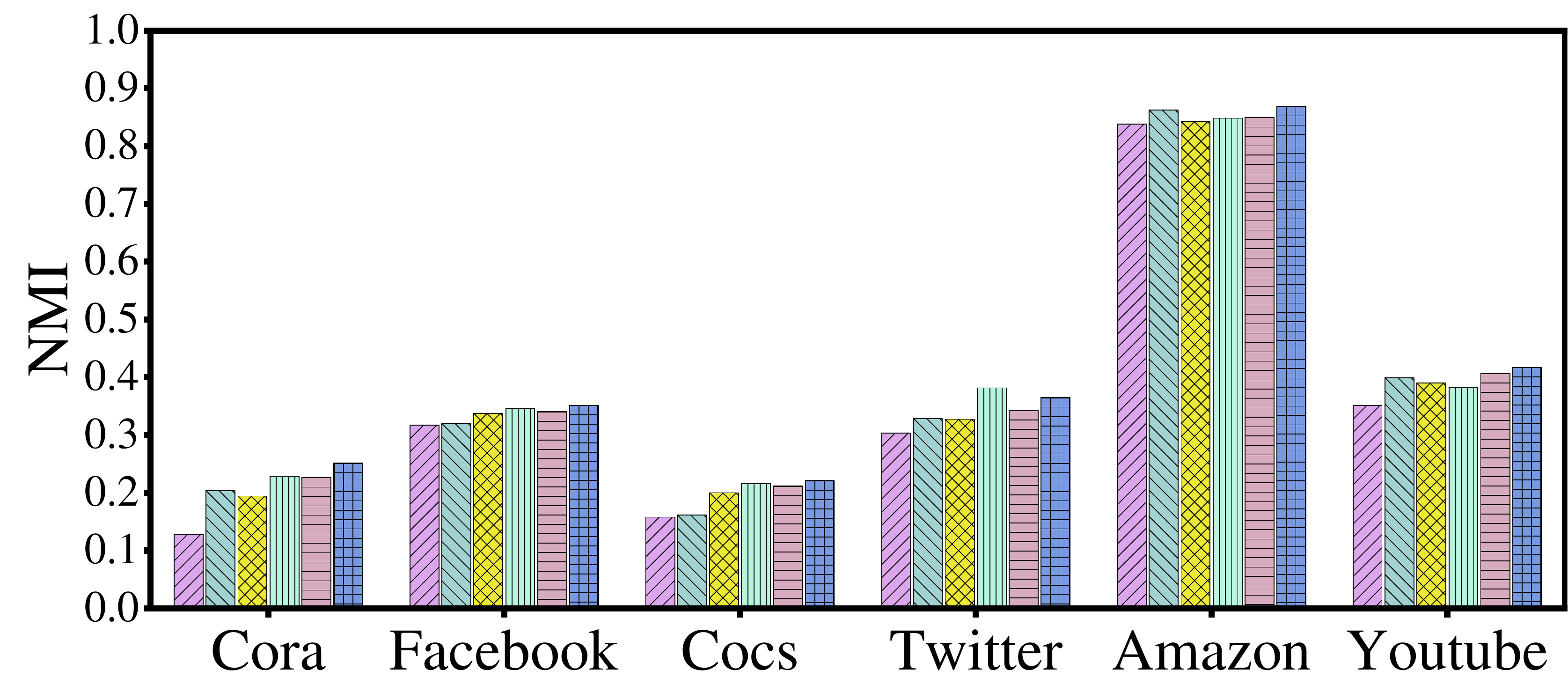}
    }
    \subfigure{
    \includegraphics[width=0.3\linewidth]{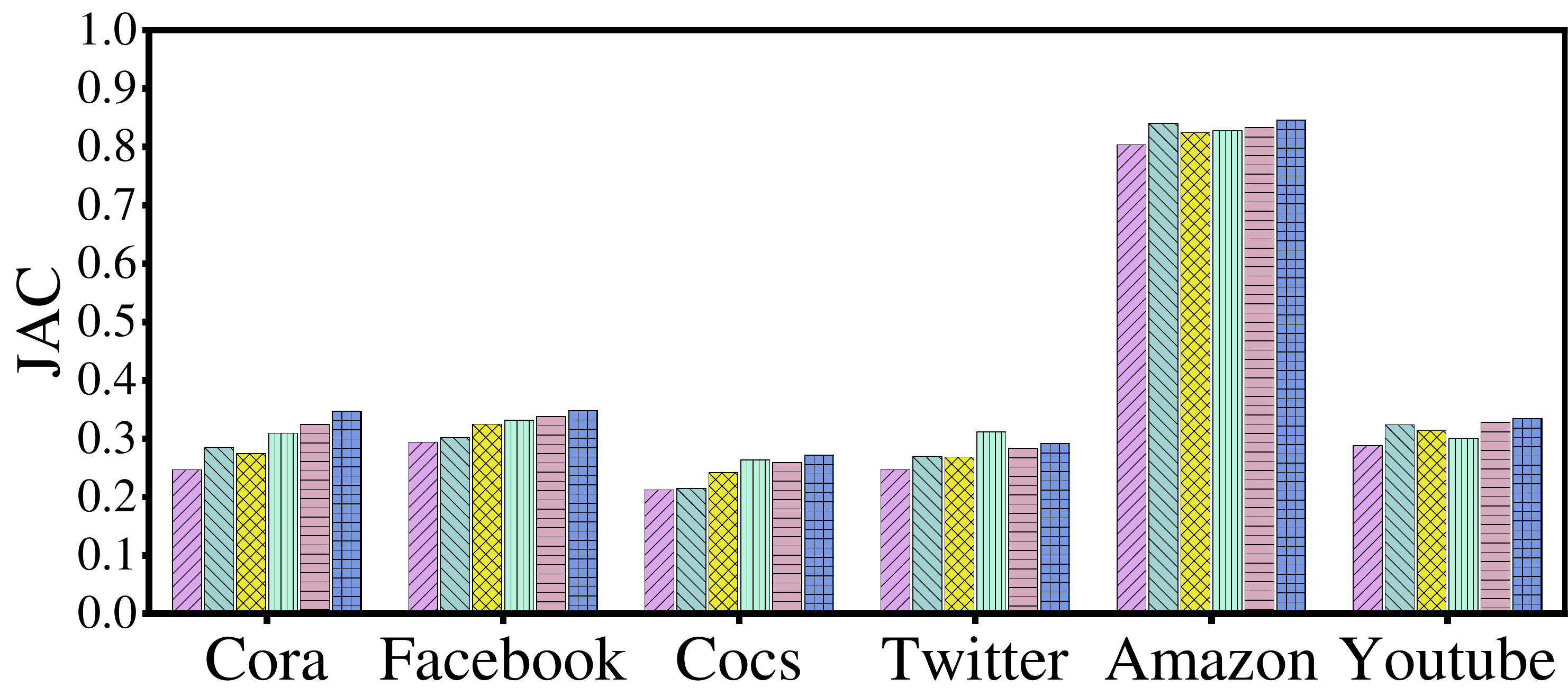}
    }
    \vspace{-0.2cm}
    \caption{Ablation studies (best view in color).} \vspace{-0.3cm}
    \label{fig:ablation}
\end{figure*}

\begin{figure*}[t!]
    \centering
    \subfigure[the ground-truth community]{
    \includegraphics[width=0.23\linewidth]{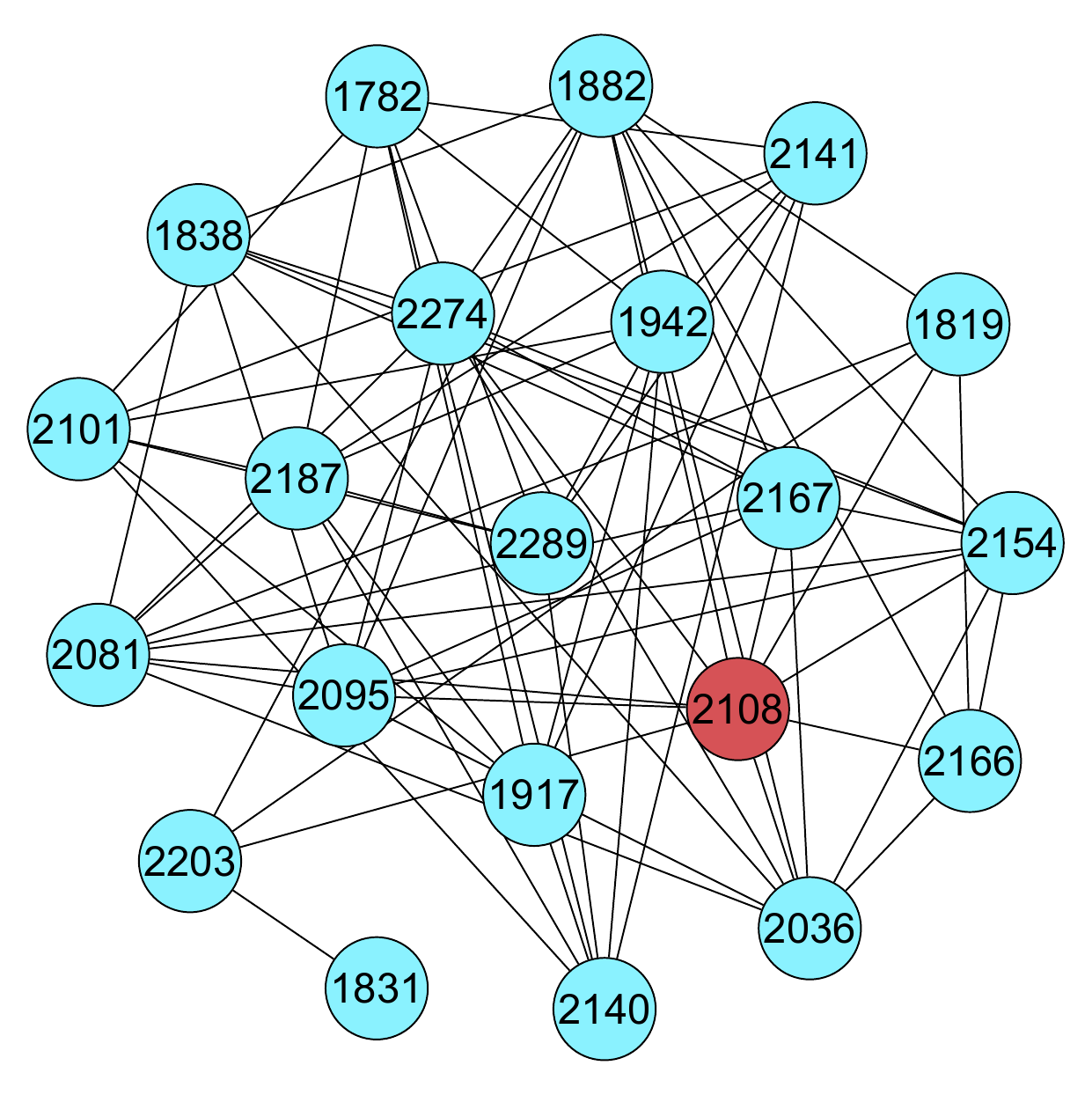}
    }
    \subfigure[CommunityAF]{
    \includegraphics[width=0.23\linewidth]{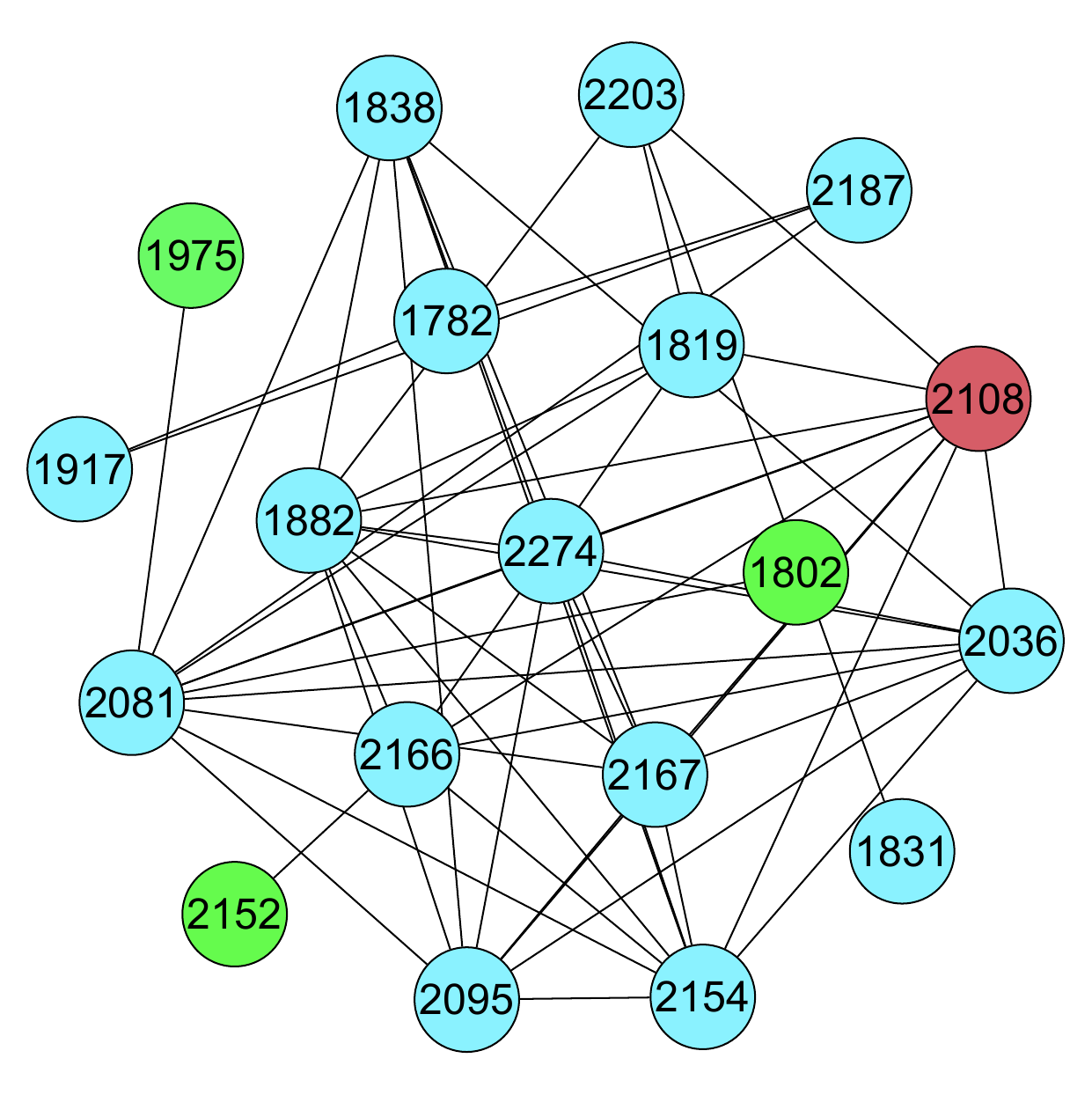}
    }
    \subfigure[TransZero]{
    \includegraphics[width=0.23\linewidth]{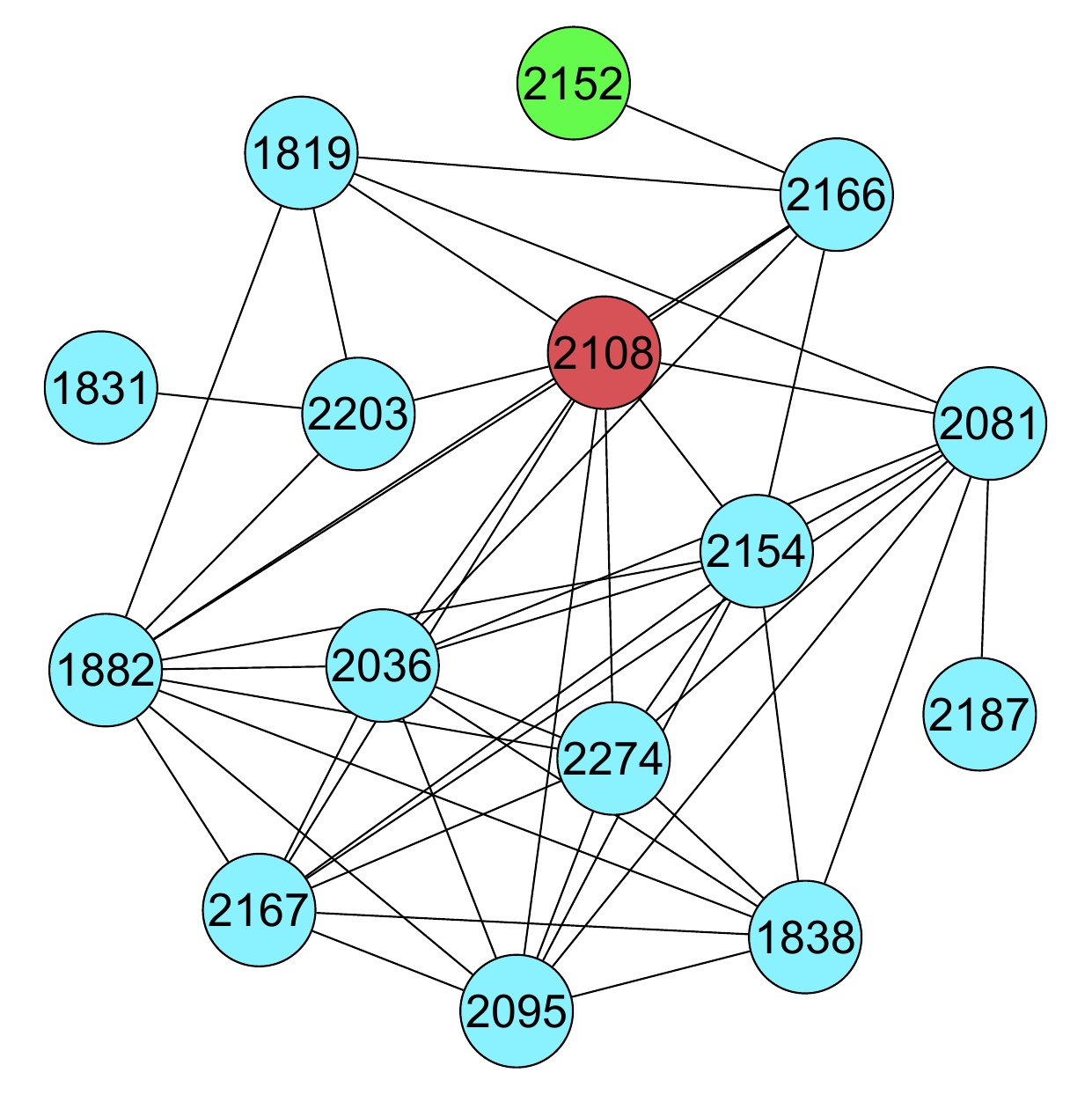}
    }
        \subfigure[NCSAC]{
    \includegraphics[width=0.23\linewidth]{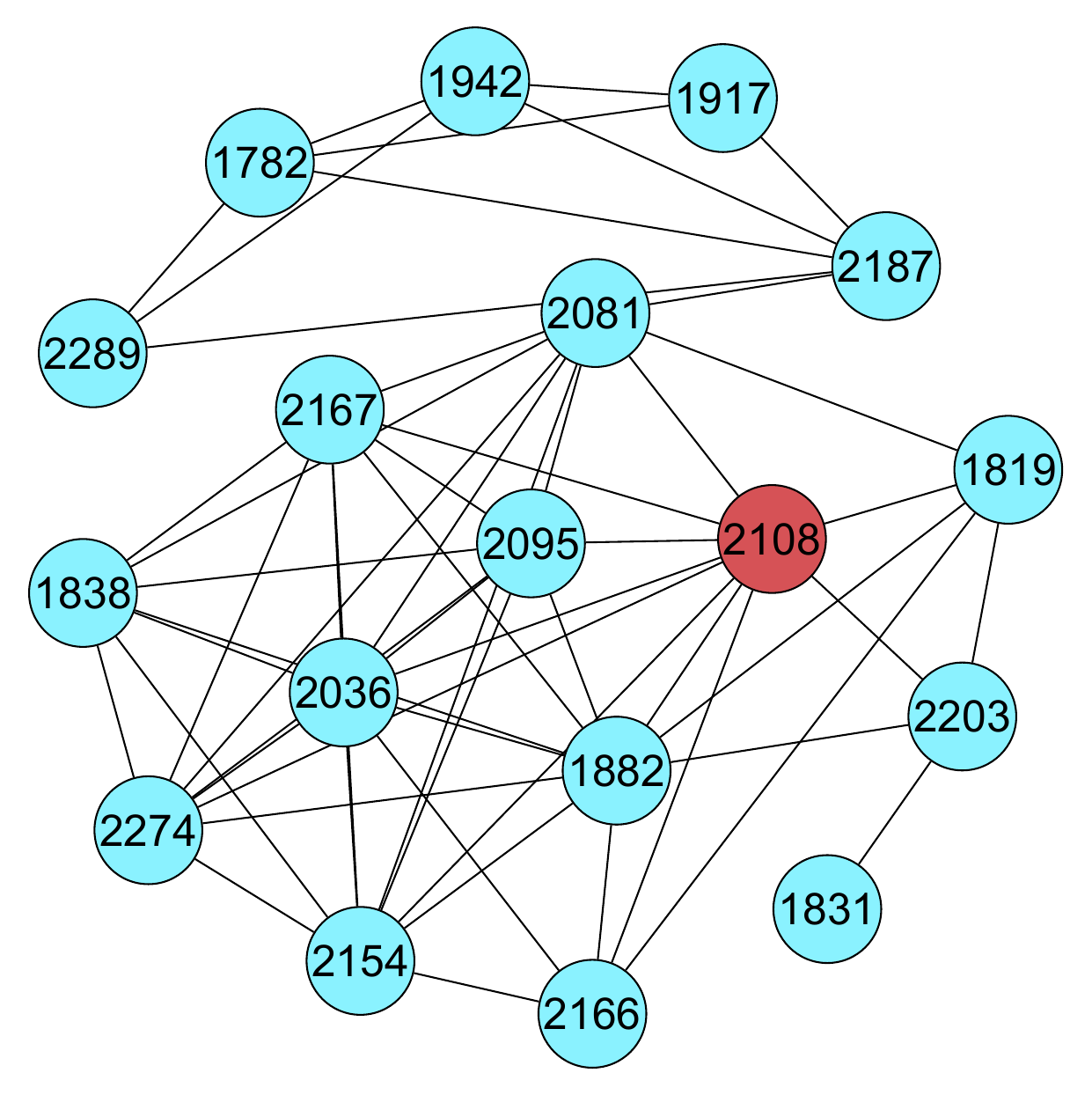}
    }
    \vspace{-0.2cm}
    \caption{Case studies. The red node represents the query node, while the blue node belongs to the ground-truth community and the green node does not.} \vspace{-0.3cm}
    \label{fig:case}
    \vspace{-0.3cm}
\end{figure*}

\stitle{Exp-7: Ablation studies.} Figure \ref{fig:ablation} illustrates the ablation studies of four key components of our \ourname. Specifically, NCSAC w/o Refiner is the NCSAC without the community refinement stage, NCSAC w/o r\_label denotes training the two policy networks without using the reward enhancement mechanism $r_{label}$, NCSAC w/o L\_C involves training the state encoder without incorporating the contrastive loss $L_C$, NCSAC w/o L\_T trains the state encoder without the triplet loss $L_T$,  NCSAC-PG substitutes the PPO algorithm with a traditional policy gradient (PG) method. Note that we only explain the F1-score because NMI and JAC have similar results. For the F1-score,  we have the following conclusions: (1) \textit{NCSAC vs. NCSAC w/o Refiner.} The community refinement module significantly improves the quality of the candidate coarse community, fulfilling its intended purpose. Specifically, the refinement module improved by at least 0.05 across all datasets, with the most significant improvement reaching up to 0.1 on the Cora. (2) \textit{NCSAC vs. NCSAC w/o r\_label.} The reward enhancement mechanism $r_{label}$ is designed to facilitate the convergence of the policy networks, leading to consistent improvements in community quality across all six datasets.   Notably, it achieves a gain of 0.08 on Cocs and 0.01 on Amazon, which aligns well with our design intuition.   The effectiveness of this mechanism becomes particularly evident in large communities, where the $\textit{Obj}(\cdot)$ (Eq.~\eqref{reward}) change caused by adding a single node is marginal, making convergence difficult.   In such cases,$r_{label}$ provides additional guidance that stabilizes the training process. (3) \textit{NCSAC vs.  NCSAC w/o L\_C \& NCSAC w/o L\_T.} Within the community-aware state encoder (Section \ref{subsec:ce}), the contrastive loss and triplet loss can indeed effectively incorporate community awareness into the embedded space of nodes (Eq. \ref{eq:loss}). In particular, triplet loss $L_T$ plays a more substantial role in enhancing state representations compared to contrastive loss $L_C$. For example, on Cora, the F1-score of NCSAC w/o L\_T is 0.41, while that of  NCSAC w/o L\_C is 0.49. (4) \textit{NCSAC vs. NCSAC-PG.} In the reinforcement learning component, the proximal policy optimization (as outlined in Eq. \ref{PPO}) demonstrates an average improvement of 0.2 compared to PG across five of the six datasets. Nonetheless, on the Twitter dataset, PG outperforms PPO by approximately 0.1, which we postulate might be attributed to overfitting in PPO under the same sampling conditions. (5) The full set of components (i.e., \ourname) yields the best performance, indicating the effectiveness of these components in enhancing algorithm performance. In summary, these results provide some preliminary evidence that the four key components we have designed are practical and effective in real-world applications.

\stitle{Exp-8: Case studies.} We conduct  case studies on Facebook \cite{gt} with the ground-truth community to further evaluate the utility of the proposed NCSAC. Specifically, we choose node 2108 as the query node (marked in red), and the communities identified by different models are visualized in Figure \ref{fig:case}. As can be seen, the communities identified by both CommunityAF and TransZero include nodes (depicted in green) that do not belong to the ground-truth community.  Besides, they also miss many nodes within the ground-truth community. Consequently, their F1-scores are 0.8 and 0.77, respectively. In contrast,  our NCSCA exclusively identifies nodes from the ground-truth community,  with only three nodes from the ground-truth community being  undetected. Thus, NCSCA achieves the highest congruence with the ground-truth community, boasting an impressive F1-score of 0.92. This result can be explained as follows: TransZero employs a  stricter termination condition, halting the search once the community quality starts to diminish. While this prevents over-expansion, it also restricts the exploration of potentially relevant nodes. CommunityAF, on the other hand, utilizes a sliding window mechanism, enabling it to explore a broader range of promising nodes but at the expense of including more irrelevant ones. However, our NCSCA adopts the more flexible termination strategy (details in Section \ref{subsec:eff}), which allows it to align closely with the ground-truth community.

\vspace{-0.2cm}
\section{Related Work}
\vspace{-0.2cm}\label{sec:existing}



\stitle{Rule-based Community Search.} 
Community search has emerged as an essential instrument for revealing the local structures within networks \cite{fang2020survey}.  Numerous models have been proposed in the literature, which can be categorized into two main types: subgraph-based and optimization-based. Specifically, $k$-core \cite{DBLP:journals/dase/LiZSYG24,k-core1,k-core2,k-core3,k-core4,DBLP:journals/pvldb/LinYLZQJJ24,DBLP:conf/dasfaa/ZhangLYJ22}, $k$-truss \cite{k-truss1, k-truss3,DBLP:journals/corr/abs-2410-15046}, and $k$-ecc \cite{k-ecc} are the representative subgraph-based models, which restrict the local connectivity of nodes or edges. For example, $k$-core is a subgraph such that each node has at least $k$-neighbors within the subgraph. On the other hand, density \cite{DBLP:journals/pvldb/WuJLZ15,DBLP:conf/sigmod/DaiQC22,DBLP:conf/kdd/YeLLLLW24}, modularity \cite{DMCS}, and conductance \cite{DBLP:conf/focs/AndersenCL06,DBLP:conf/kdd/KlosterG14,DBLP:conf/aaai/LinLJ23,DBLP:journals/eswa/HeLYLJW24} are the representative optimization-based models, which aim to achieve a globally optimal objective function. However, these rule-based community search models suffer  from the following defects: (1) They   (excluding conductance) focus on internal cohesion and neglects external sparsity, which violates the natural intentions of the community \cite{DBLP:journals/pacmmod/KamalB24,DBLP:journals/eswa/HeLYLJW24,DBLP:conf/www/LeskovecLM10}. (2) Conductance is defined based on structure, thus neglecting important node attributes. (3) They are plagued by structural inflexibility, wherein the ground-truth community may not align with user-predefined constraints, thereby resulting in suboptimal community quality \cite{ICS-GNN,QD-GNN,CommunityAf,coclep,cgnp,TransZero}.


\stitle{Learning-based Community Search.}
With the swift progress of artificial intelligence, researchers have commenced the integration of Graph Neural Network (GNN) into community search. For instance, \emph{ICS-GNN} \cite{ICS-GNN} pioneered the GNN-based approach to community search by iteratively identifying the target community with maximum GNN scores after crawling online subgraphs. \emph{QD-GNN} \cite{QD-GNN} combined local query dependency structures with global graph embedding. \emph{CommunityAF} \cite{CommunityAf} utilized autoregressive flow to generate communities and can learn various community patterns. \emph{CGNP} \cite{cgnp} and \emph{COCLEP} \cite{coclep} leveraged meta-learning and contrastive learning to solve the community search problem, respectively. However, all of these approaches completely overlook the traditional rule-based  constraints. Recently,  \emph{ALICE} \cite{ALICE} and \emph{TransZero} \cite{TransZero} effectively integrate traditional rule-based subgraph metrics with GNN, resulting in substantial improvements. At a high level, they first extracted a predefined new subgraph $G_{new}$, then identified the target community by executing the node classification task within $G_{new}$. Unfortunately, 
they are heavily dependent on the performance of $G_{new}$,
since nodes outside $G_{new}$ have no chance to contribute to the target community, resulting in poor quality (Section 6). 
However, our method focuses on optimizing communities by strategically removing and adding nodes, starting from $G_{new}$.

\vspace{-0.3cm}
\section{Conclusion}
In this paper, we propose a novel Neural Community Search via Attribute-augmented Conductance, termed \ourname,  to integrate deep learning techniques with traditional rule-based constraints, thereby improving the quality of community search. Specifically, \ourname \ utilizes the proposed attribute-augmented conductance to adaptively extract a coarse candidate community with satisfactory quality. Subsequently, \ourname \ refines the candidate community through advanced reinforcement learning techniques, ultimately producing high-quality results.  Extensive experiments are conducted on six real-life graphs, in comparison with ten competitors, to demonstrate the superiority of our proposed solution.  The following two aspects of future work are worth exploring. (1) Given that real-world graphs frequently exhibit characteristics such as dynamism, heterogeneity, heterophily, or even noise, exploring ways to extend our algorithm to accommodate these more intricate scenarios represents a promising research direction. (2) We aim to devise multiple parallel strategies to bolster the scalability of our algorithms.

\bibliographystyle{IEEEtran}
\bibliography{main}

\end{document}